\documentclass{article}
\pdfoutput=1

\usepackage[preprint, nonatbib]{neurips_2024}

\usepackage[utf8]{inputenc} 
\usepackage[T1]{fontenc}    
\usepackage{hyperref}       
\usepackage{url}            
\usepackage{booktabs}       
\usepackage{amsfonts}       
\usepackage{nicefrac}       
\usepackage{microtype}      
\usepackage{xcolor}         
\usepackage{bm} 
\usepackage{wrapfig}
\usepackage{amsmath}
\usepackage{amssymb}
\usepackage{mathtools}
\usepackage{amsthm}
\usepackage{graphicx}
\theoremstyle{plain}

\newtheorem{proposition}{Proposition}

\theoremstyle{definition}

\theoremstyle{remark}

\usepackage{algorithm} 
\usepackage{algorithmic}
\usepackage{multicol}
\usepackage{multirow}
\usepackage{xcolor}
\definecolor{lightgray}{rgb}{0.9, 0.9, 0.9} 
\title{Knowing What Not to Do: Leverage Language Model Insights for Action Space Pruning in Multi-agent Reinforcement Learning}

\author{%
  Zhihao Liu\textsuperscript{1,4}$^*$, Xianliang Yang\textsuperscript{2},  Zichuan Liu\textsuperscript{3}, Yifan Xia\textsuperscript{3},\\
  \textbf{Wei Jiang\textsuperscript{5},  Yuanyu Zhang\textsuperscript{6}, Lijuan Li\textsuperscript{1}, Guoliang Fan\textsuperscript{1},  Lei Song\textsuperscript{2},  Jiang Bian\textsuperscript{2}$^\dagger$}\\
  \textsuperscript{1}Institute of Automation, Chinese Academy of Sciences, \textsuperscript{2}Microsoft Research Asia,\\
  \textsuperscript{3}Nanjing University,  \textsuperscript{4}School of Artificial Intelligence, University of Chinese Academy of Sciences,\\
  \textsuperscript{5}University of Illinois Urbana-Champaign, \textsuperscript{6}Guizhou University\\
  \texttt{\{liuzhihao2022, lijuan.li, guoliang.fan\}@ia.ac.cn},\\ \texttt{\{zichuanliu, yfxia\}@smail.nju.edu.cn}, \texttt{zhangyuanyu18@gmail.com}\\
  \texttt{\{xianya, lei.song, jiang.bian\}@microsoft.com}, \texttt{weij4@illinois.edu}\\
}

\newcommand\nnfootnote[1]{%
  \begin{NoHyper}
  \renewcommand\thefootnote{}\footnote{#1}%
  \addtocounter{footnote}{-1}%
  \end{NoHyper}
}

\begin{document}

\nnfootnote{$^*$ Work done during internship at Microsoft Research Asia.}
\nnfootnote{$\dagger$ Corresponding author.}

\maketitle

\begin{abstract}\label{abstract}
  Multi-agent reinforcement learning (MARL) is employed to develop autonomous agents that can learn to adopt cooperative or competitive strategies within complex environments. However, the linear increase in the number of agents leads to a combinatorial explosion of the action space, which may result in algorithmic instability, difficulty in convergence, or entrapment in local optima. While researchers have designed a variety of effective algorithms to compress the action space, these methods also introduce new challenges, such as the need for manually designed prior knowledge or reliance on the structure of the problem, which diminishes the applicability of these techniques.
  In this paper, we introduce \textbf{E}volutionary action \textbf{SPA}ce \textbf{R}eduction with \textbf{K}nowledge (\texttt{eSpark}), an exploration function generation framework driven by large language models (LLMs) to boost exploration and prune unnecessary actions in MARL. Using just a basic prompt that outlines the overall task and setting, \texttt{eSpark} is capable of generating exploration functions in a zero-shot manner, identifying and pruning redundant or irrelevant state-action pairs, and then achieving autonomous improvement from policy feedback. In reinforcement learning tasks involving inventory management and traffic light control encompassing a total of 15 scenarios, \texttt{eSpark} consistently outperforms the combined MARL algorithm in all scenarios, achieving an average performance gain of 34.4\% and 9.9\% in the two types of tasks respectively. Additionally, \texttt{eSpark} has proven to be capable of managing situations with a large number of agents, securing a 29.7\% improvement in scalability challenges that featured over 500 agents. The code can be found in \url{https://github.com/LiuZhihao2022/eSpark.git}.
\end{abstract}

\section{Introduction}

\label{Introduction}
Multi-agent reinforcement learning (MARL) has emerged as a powerful paradigm for solving complex and dynamic problems that involve multiple decision-makers~\cite{zhang2021multi, wang2020qplex}. 
However, the intricacies of agent interplay and the exponential expansion of state and action spaces render the solution of MARL problems difficult. Researchers have proposed the Centralized Training with Decentralized Execution (CTDE) framework~\cite{oliehoek2008optimal} and parameter sharing methods, decomposing the value or policy functions of a multi-agent system into individual agents and sharing model parameters among all agents. As experimentally verified by many of the most prominent MARL algorithms such as Multi-agent PPO (MAPPO)~\cite{yu2022surprising}, QMIX~\cite{rashid2020monotonic}, QPLEX~\cite{wang2020qplex} or QTRAN~\cite{son2019qtran}, these methodologies have been demonstrated to be robust strategies for surmounting the challenges posed by MARL. MARL methods based on parameter sharing and CTDE have achieved notable success in a variety of well-established tasks, including StarCraft Multi-Agent Challenge (SMAC)~\cite{li2023explicit, wang2020roma}, the Multi-Agent Particle Environment (MPE)~\cite{lowe2017multi}, and  Simulation of Urban MObility (SUMO)~\cite{wei2019colight,lu2023dualight}.

Despite the great success of parameter-sharing CTDE methods, their practicality dwindles in real-world tasks involving large number of agents, such as large-scale traffic signal control~\cite{mousavi2017traffic}, wireless communication networks~\cite{zocca2019temporal}, and humanitarian assistance and disaster response~\cite{meier2015digital}, where centralized training becomes impractical due to large problem scale~\cite{munir2021multi}. Fully Decentralized Training and Execution (DTDE) methods, such as Independent PPO (IPPO)~\cite{de2020independent}, offer a scalable solution where resource consumption does not escalate drastically with an increase in the number of agents. However, due to the lack of consideration for agent interactions, they often struggle to find optimal solutions and fall into local optima. Current strategies for addressing large-scale MARL tasks involve introducing task-specific structures to model agent interactions or dividing agents into smaller, independently trained groups~\cite{ying2023scalable, chen2020toward}. These methods, however, are constrained by their dependence on task-related structuring, limiting their applicability to a narrow range of problems.

Additionally, the "curse of dimensionality" presents a significant challenge in multi-agent systems~\cite{hao2022breaking, jianye2022boosting}, where agents are required to navigate through an expansive action space saturated with numerous actions that are either irrelevant or markedly suboptimal (relative to states). While humans can deftly employ contextual cues and prior knowledge to sidestep such challenges, MARL algorithms typically engage in the exploration of superfluous and extraneous suboptimal actions~\cite{Zahavy_Haroush_Merlis_Mankowitz_Mannor_2018}. Besides, prevailing parameter sharing can exacerbate this exploratory dilemma, as will be elucidated in Proposition~\ref{proposition 2}. 
The issue occurs primarily because agents with shared parameters often prefer suboptimal policies that present short-term advantages, rather than exploring policies that may potentially deliver higher long-term returns.

Exploration is crucial for overcoming local optima, as it encourages agents to discover potentially better states thus refining their policies. While single-agent exploration techniques like the Upper Confidence Bound (UCB)~\cite{auer2002using}, entropy regularization~\cite{haarnoja2018soft}, and curiosity-based exploration~\cite{groth2021curiosity, pathak2017curiosity} have shown promising results, they struggle with the escalated complexity in MARL scenarios, compounded by issues like deceptive rewards and the "Noisy-TV" problem~\cite{burda2018exploration}. Integrating domain knowledge into exploration could significantly enhance exploration efficiency, by helping identify critical elements and problem structures, thereby aiding in the selection of optimal actions~\cite{simon1956rational}. However, the integration of this knowledge within a data-driven framework poses significant challenges, particularly when manual input from domain experts is required, thus reducing its practicality.

Recently, Large Language Models (LLMs) such as GPT-4~\cite{achiam2023gpt} have shown formidable skills in language comprehension, strategic planning, and logical reasoning across various tasks~\cite{yao2022react,zhu2023ghost}.
Although not always directly solving complex, dynamic problems, their inferential and error-learning abilities facilitate progressively better solutions through iterative feedback~\cite{ma2023eureka}. 
The integration of LLMs with MARL presents a promising new avenue by facilitating exploration through the pruning of redundant actions. In this paper, we introduce \textbf{E}volutionary action \textbf{SPA}ce \textbf{R}eduction with \textbf{K}nowledge (\texttt{eSpark}), a novel approach that utilizes LLMs to improve MARL training via optimized exploration functions, which are used to prune the action space. \texttt{eSpark} begins by using LLMs to generate exploration functions from task descriptions and environmental rules in a zero-shot fashion. It then applies evolutionary search within MARL to pinpoint the best performing policy. Finally, by analyzing the feedback on the performance of this policy, \texttt{eSpark} reflects and proposes a set of new exploration functions, and iteratively optimizes them according to the aforementioned steps. This process enhances the MARL policy by continuously adapting and refining exploration. To summarize, our contributions are as follows: 

\begin{enumerate}
    \item We introduce the \texttt{eSpark} framework, which harnesses the intrinsic prior knowledge and encoding capability of LLMs to design exploration functions for action space pruning, thus guiding the exploration and learning process of MARL algorithms. \texttt{eSpark} requires no complex prompt engineering and can be easily combined with MARL algorithms.
    \item We validate the performance of \texttt{eSpark} across 15 different environments within the inventory management task MABIM~\cite{yang2023versatile} and the traffic signal control task SUMO~\cite{behrisch2011sumo}. Combined with IPPO, \texttt{eSpark} outperforms IPPO in all scenarios, realizing an average profit increase of 34.4\% in the MABIM and improving multiple metrics in SUMO by an average of 9.9\%. Even in the face of scalability challenges where the DTDE methods typically encounter limitations, \texttt{eSpark} elevates the performance of IPPO by 29.7\%.
    \item 
    We conduct controlled experiments and ablation studies to analyze the effectiveness of each component within the \texttt{eSpark} framework. We first validate the advantages of knowledge-based pruning. Subsequently, we conduct ablation studies to demonstrate that both retention training and LLM pruning techniques contribute to the performance of \texttt{eSpark}. These effects are even more pronounced in the more complex MABIM environment.
\end{enumerate}

\section{Related works}
\textbf{LLMs for code generation.} LLMs have exhibited formidable capabilities in the domain of code generation~\cite{roziere2023code, nejjar2023llms}. More recently, LLM-based approaches to self-improving code generation have been applied to address challenges in combinatorial optimization~\cite{liu2024example, ye2024reevo}, robotic tasks~\cite{liang2023code, wang2024demo2code}, reinforcement learning (RL) reward design~\cite{ma2023eureka}, and code optimization~\cite{romera2024mathematical}. For the first time, to the best of our knowledge, we propose the use of LLMs to generate code with the aim of pruning the redundant action space in MARL environments, and through a process of autonomous reflection and evolution, iteratively enhancing the quality of the generated output.

\textbf{LLMs for RL and MARL.} The integration of LLMs into RL and MARL has sparked considerable research interest~\cite{sharma2022skill, kwon2023reward, li2024verco}. Some works~\cite{hill2020human, chan2019actrce} incorporate the goal descriptions of language models that help in enhancing the generalization capabilities of agents designed to follow instructions. Further studies have extended this approach to complex tasks involving reasoning and planning~\cite{huang2023inner, huang2022language}. Moreover, LLMs have been employed to guide exploration and boost RL efficiency~\cite{du2023guiding, chang2023learning, hu2023language}. However, scaling to high-complexity, real-time, multi-agent settings remains a challenge. Our method mitigates this by generating exploration functions to navigate the policy space, thus facilitating the application to complex MARL scenarios without direct LLM-agent decision-making interaction.

\textbf{Action space pruning in RL and MARL.} 
Pruning the action space has been shown to be effective in guiding agent behaviors in complex environments~\cite{lipton2016combating,fulda2017can}. Techniques include learning an elimination signal to discard unnecessary actions~\cite{Zahavy_Haroush_Merlis_Mankowitz_Mannor_2018}, and employing transfer learning that pre-trains agents to isolate useful experiences for later action refinement~\cite{shirali2023pruning,lan2022transfer,ammanabrolu2018playing}. Some approaches use manually designed data structures based on prior knowledge to filter actions~\cite{dulac2015deep,padullaparthi2022falcon,Nagarathinam_Menon_Vasan_Sivasubramaniam_2020}. However, training pruning signals or applying transfer learning is inherently difficult with many agents, and the need for expert knowledge in manual pruning rules hampers their transferability, limiting the applicability of current methods. We harness the abundant knowledge embedded within LLMs for action space pruning, demonstrating universal applicability across a multitude of scenarios.

\section{Preliminaries}
\subsection{Problem formulation and notations}\label{problem formulation}
\paragraph{Markov game framework.} In our study, we explore a Markov game framework, formally defined by the tuple $\left \langle N, \mathcal{S}, \bm{\mathcal{O}}, \bm{\mathcal{A}}, P, R, \gamma \right \rangle$. Here $N$ represents the total number of participating agents, $\mathcal{S}$ denotes a well-defined state space, and $\bm{\mathcal{O}}=\prod_{k=1}^N\mathcal{O}^k$ constitutes the combined observation space, $\bm{\mathcal{A}}=\prod_{k=1}^N\mathcal{A}^k$ is the joint action space for all agents involved. The transition dynamics are captured by the probability function $P:\mathcal{S}\times\bm{\mathcal{A}}\times\mathcal{S}\to [0,1]$, the reward function $R:\mathcal{S}\times\bm{\mathcal{A}}\to \mathbb{R}$ maps state-action pairs to real-valued rewards. The discount factor is denoted by $\gamma\in [0,1]$.

At each discrete time step $t$, the environment is in state $s_t\in \mathcal{S}$. Each agent $k\in [1,2,\dots,N]$ receives an observation $o_t^k\in \mathcal{O}^k$ and draws an action from $a_t^k\sim\pi^k(\cdot\mid o_t^k)$, where $\pi^k:\mathcal{O}^k\times\mathcal{A}^k\rightarrow [0,1]$ denotes the policy of agent $k$, and $\sum_{a^k\in\mathcal{A}^k}\pi^k(a^k \mid  o_t^k)=1$. The joint actions of all agents $\textbf{a}_t=(a_t^1,a_t^2,\ldots,a_t^N)$ is drawn from the joint policy $\bm{\pi}(\cdot\mid s_t)=\prod_{k=1}^N\pi^k(\cdot\mid o_t^k)$. Subsequently, a reward ${\rm r}_t=R(s_t,\textbf{a}_t)$ is given based on the current state and joint action. The state transition is determined by $s_{t+1} \sim P(\cdot\mid s_t,\textbf{a}_t)$.

In this paper, we focus on a fully cooperative scenario where all agents share a common reward signal. The collective objective is to maximize the expected cumulative reward, starting from an initial state distribution $\rho^0$. This collaborative approach emphasizes the alignment of individual agent strategies towards maximizing a unified reward $J(\bm{\pi})$:
\[
    J(\bm{\pi})=\mathbb{E}_{s_0\sim \rho^0, \textbf{a}_{0:\infty}\sim \bm{\pi}, s_{1:\infty}\sim P}
    \left[\sum_{t=0}^\infty {\rm \gamma^t r_t}\right].
\]

\paragraph{Policy with exploration function.}In the configuration of our study, we introduce the exploration function $E: \mathcal{O}^k \times \mathcal{A}^k \rightarrow \{0,1\}$, indicating whether an action is selectable by agent $k$. For a given policy $\pi^k$ of agent $k$ and an exploration function $E$, we define a new policy guided by $E$, denoted by $\pi_E^k$, as follows:
\[
\pi_E^k(\cdot\mid o_t^k)=\frac{\pi^k(\cdot \mid  o_t^k)\cdot E(o_t^k, \cdot)}{\sum_{a^k\in\mathcal{A}^k}\pi^k(a^k \mid  o_t^k)\cdot E(o_t^k, a^k)}
\]
if $\sum_{a^k\in\mathcal{A}^k}\pi^k(a^k \mid  o_t^k)\cdot E(o_t^k, a^k) > 0$; otherwise, $\pi_E^k(\cdot\mid o_t^k)=\pi^k(\cdot \mid  o_t^k)$. Consequently, the joint policy for all agents under the guidance of $E$ is defined as: 
\[
\boldsymbol{\pi}_E(\cdot\mid s_t) = \prod_{k=1}^N \pi_E^k(\cdot\mid o_t^k).
\]

We define the set of all joint policies as $\{\bm{\pi}\}$ and the set of all exploration functions as $\{E\}$. Let $\{\bm{\pi}_E\}$ denote the set of joint policies when subjected to an exploration function $E \in \{E\}$. An exploration function $E$ is non-trivial if it assigns a zero probability to at least one observation-action pair. The following proposition naturally arises from the definition: 
\begin{proposition} \label{proposition 1}
\begin{enumerate}
    \item[] 
    \item For any $E \in \{E\}$, $\{\bm{\pi}_E\}\subseteq\{\bm{\pi}\}$. If $E$ is non-trivial, then $\{\bm{\pi}_E\}\subset\{\bm{\pi}\}$.
    \item For any $\bm{\pi} \in \{\bm{\pi}\}$, there exists a non-trivial $E\in\{E\}$ such that $J(\bm{\pi}_{E}) \geq J(\bm{\pi})$.
\end{enumerate}
\end{proposition}

An intelligent choice of exploration functions does not diminish our ability to discover optimal policies; instead, it allows us to refine the policy space, thereby enhancing the efficiency of the learning process. The proof of this proposition can be found in Appendix~\ref{appendix 1}.

\subsection{Challenges and motivations}
The intricate relationships among multiple agents make it extremely difficult to search for the optimal solution in MARL. Without powerful exploration methods, it is nearly impossible to avoid suboptimal outcomes.
We will elaborate on this with an example from the following proposition:

\begin{proposition} \label{proposition 2}
    Let's consider a fully cooperative game with N agents, one state, and the joint action space $\{0,1\}^N$, where the reward is given by $r({\bm{\mathrm{0}}^0, \bm{\mathrm{1}}^N})=r_1$ and $r({\bm{\mathrm{0}}^{N-m},\bm{\mathrm{1}}^m})=-mr_2$ for all $m\neq N$, $r_1$, $r_2$ are positive real numbers. We suppose the initial policy is randomly and uniformly initialized, and the policy is optimized in the form of gradient descent. Let $p$ be the probability that the shared policy converges to the best policy, then:

    \[
        p=1-\sqrt[N-1]{\frac{r_2}{r_1+Nr_2}}.
    \]

\end{proposition}
Detailed proof is provided in Appendix~\ref{appendix 1}. In the above example, we show that the increase in the number of agents makes it more difficult for MARL algorithms to reach the optimal solution. 
However, based on the problem context, humans can understand problems from a high-level semantic perspective, and quickly find optimal solutions. As LLMs have demonstrated surprising abilities in semantic understanding, reasoning, and planning across various tasks~\cite{yuan2023plan4mc, wang2023voyager}, we conduct a simple experiment to test GPT-4's capability for the issue in Proposition~\ref{proposition 2}, and here is GPT-4's response:

\colorbox{lightgray}{%
  \begin{minipage}{\textwidth} 
    \texttt{\footnotesize In a fully cooperative game, all agents work together to maximize the total reward. 
    There are two distinct reward conditions: \begin{enumerate}  
    \item When all agents choose action 1, the reward is $r_1$, a
    positive real number.  
    \item When there is any number of agents $m$ (where $0<m<N$) choosing action 1, the reward is $-mr_2$, where $r_2$ is a positive real number.  
    \end{enumerate} 
    All agents should act in a way that avoids the negative reward scenario. The negative reward scenario happens anytime there is a mix of 0’s and 1's in the action space, which means some agents are choosing 1 and others are choosing 0. Therefore, the optimal joint action for all agents is to all choose 1.}
  \end{minipage}%
}

GPT-4 exhibits reasoning abilities on par with those of humans and directly solves the problem in Proposition~\ref{proposition 2}. Propositions~\ref{proposition 1} has already shown that an intelligent exploration function can not only reduce the searching space but also improve the final performance. 
This makes us think about applying the powerful LLMs to prune the redundant action space and thereby guide the exploration in MARL. In the following sections, we propose the \texttt{eSpark} framework, which integrates the prior knowledge and inferential capability of LLMs to boost the exploration in MARL.

\section{Method}\label{method}
In this section, we introduce a novel framework, \texttt{eSpark}, which integrates robust prior knowledge encapsulated in LLMs. It improves iteratively through a cycle of trial and error, leveraging the capability of LLMs.
Figure~\ref{eSpark_framework} illustrates the overall training procedure.
\texttt{eSpark} is composed of three components: (i) zero-shot generation of exploration functions, (ii) evolutionary search for best performing MARL policy, and (iii) detailed feedback on the performance of the policy to improve the generation of exploration functions. We show the pseudocode of $\texttt{eSpark} $in Appendix~\ref{pseudocode sec}.

\subsection{Exploration function generation}
LLMs have been demonstrated to possess exceptional capabilities in both code comprehension and generation. 
To this end, we employ a LLM as \textbf{LLM code generator}, denoted as $\texttt{LLM}_g$, whose task is to understand the objectives of the current environment, and output an exploration function:
\begin{wrapfigure}{r}{0.63\textwidth} 

  \centering  
  \includegraphics[width=0.64\textwidth]{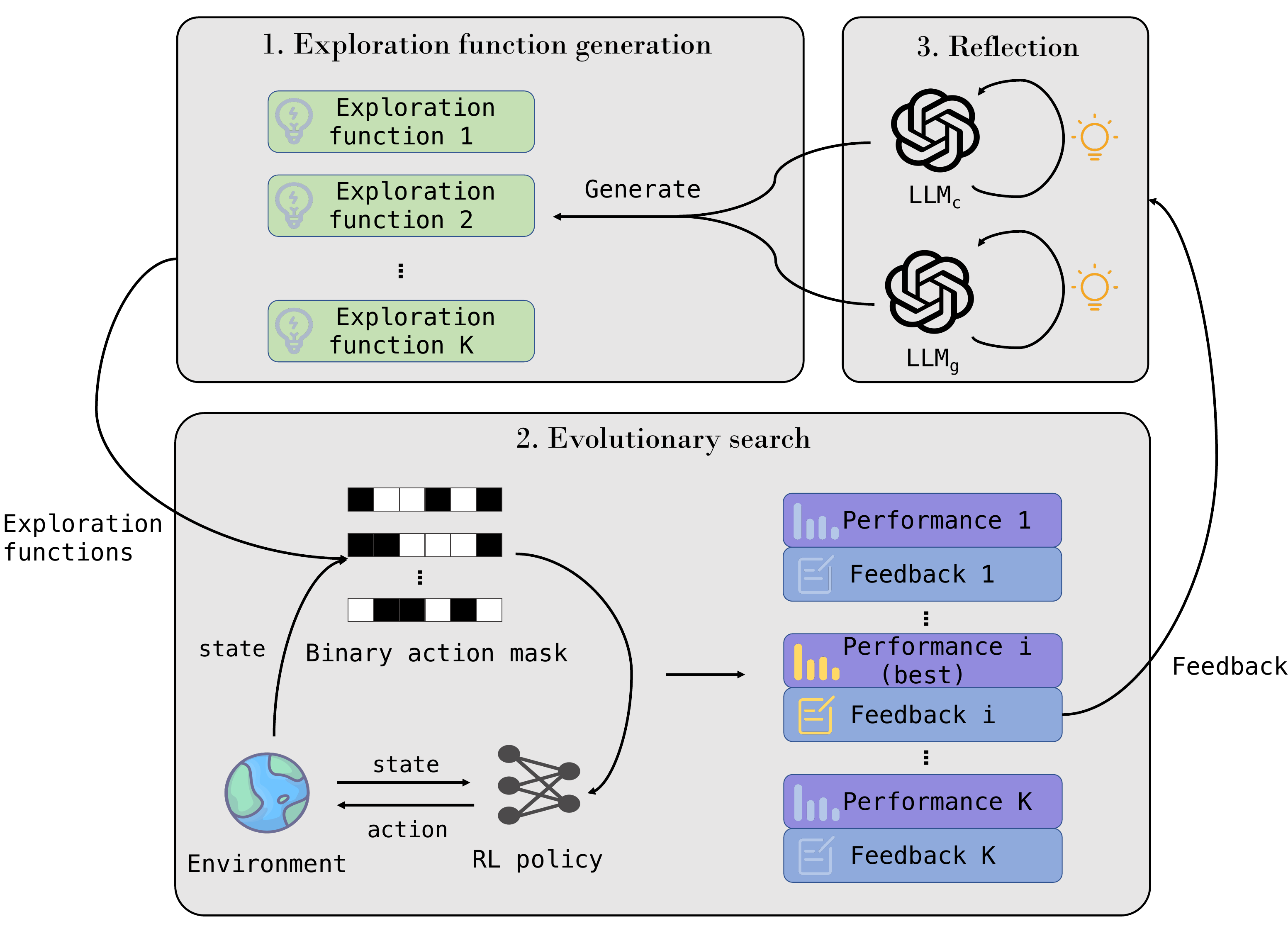}  
  \caption{\texttt{eSpark} firstly generates $K$ exploration functions via zero-shot creation. Each exploration function is then used to guide an independent policy, and the evolutionary search is performed to find the best-performing policy. Finally, \texttt{eSpark} reflects on the feedback from the best performance policy, refines, and regenerates the exploration functions for the next iteration.}
  \label{eSpark_framework}
\end{wrapfigure} 
\begin{equation}\label{eq1}
E_1, \dots, E_K \sim \texttt{LLM}_g(\texttt{prom}),
\end{equation}

where \texttt{prom} is the prompt for $\texttt{LLM}_g$, and the generation of $K$ exploration functions is to circumvent the suboptimality that may arise from single-sample generation.
The initial \texttt{prom} includes an \textit{RL formulation} describing the reward system, state items, transitions, and the action space, alongside a \textit{task description} that specifies the task objectives, expected outputs, and formatting rules. Details on the initial \texttt{prom} are provided in Appendix~\ref{full prompts}. We use code for the \textit{RL formulation} as it effectively captures the physical transition dynamics crucial to RL problems, which are always difficult to express precisely through the text alone, especially when environmental complexity increases. Code contexts also improve code generation and clarify environmental semantics and variable roles~\cite{ma2023eureka}.

During the code generation, however, $\texttt{LLM}_g$ may incorrectly interpret variables and produce logically flawed code. This kind of flawed logic could persist if it is added to the prompt context for the next generation. As research has shown that collaboration among multiple LLMs can enhance the quality and efficacy of the generated contents~\cite{chen2023agentverse,zhang2023controlling}, we introduce the $\textbf{LLM checker}$ denoted as $\texttt{LLM}_c$, which reviews $\texttt{LLM}_g$'s output to pursue an enhanced generation. $\texttt{LLM}_c$ uses the same prompt as $\texttt{LLM}_g$ but is prompted to focus on verifying the accuracy of code relative to environmental transitions and variable specifications. If inconsistencies are found, $\texttt{LLM}_c$ signals the error, prompting $\texttt{LLM}_g$ to regenerate the code. Finally, exploration functions are generated by:

\begin{equation}
    E_1, \dots, E_K \sim \texttt{LLM}_c\left(\texttt{prom}, \texttt{LLM}_g(\texttt{prom})\right).
\end{equation}

Exploration functions are applied only during the training phase to guide the exploration of MARL. During the execution phase, all exploration functions are removed.

\subsection{Evolutionary search} \label{Evolutionary Search}
During the generation, however, it should be noted that the initially generated exploration function may not always guarantee executability and effectiveness. To address this, \texttt{eSpark} performs an \textbf{evolutionary search} paradigm that selects the best-performing policy in one iteration and uses its feedback for subsequent generation~\cite{ma2023eureka}. 
Specifically, \texttt{eSpark} samples a batch of $K$ exploration functions in each generation to ensure there are enough candidates for successful code execution. 
Performance is assessed at regular checkpoints within an iteration, with the final evaluation based on the average of the last few checkpoints. 
The policy achieving the highest performance is selected, and the feedback obtained from this policy is integrated to optimize the exploration functions in the following steps.

Due to the dynamic nature of exploration, the exploration function generated based on feedback from the best-performing policy may not be applicable to other policies. As the proof of Proposition~\ref{proposition 1} demonstrates, when an exploration function is incapable of intelligently pruning, it may even impair the performance of the policy.
To this end, we utilize \textbf{retention training} to maintain continuity of exploration. 
Let $\phi_{\text{best}}^{i-1}$ represent best-performing policy from the $(i-1)$-th iteration. 
For the $i$-th iteration except for the first, at the start of the iteration, we set:
\begin{equation}
     \phi_1^i, \phi_2^i, \dots, \phi_K^i \leftarrow \phi_{\text{best}}^{i-1}.
\end{equation}
This allows us to match exploration functions with their corresponding policies, subsequently refining performance incrementally. We will verify the impact of retention training in Section~\ref{retention training result}. 

\subsection{Reflection and feedback}
Feedback from the environment can significantly enhance the quality of the generated output by LLMs~\cite{nascimento2023gpt, du2023guiding}. In \texttt{eSpark}, we leverage \textbf{policy feedback}, which contains the evaluation of policy performance from various aspects, to enhance the generation of LLMs.  This policy feedback may either come from experts or be automatically constructed from the environment, as long as it encompasses insights into the aspects where the current algorithm performs well and areas where it requires improvement. As illustrated in Equation~\ref{reflection equation}, by correlating the best-performing policy feedback $F_\text{best}$ and the most effective exploration function $E_{\text{best}}$, LLMs introspect, update the prompt $\texttt{prom}$, and gear up for the ensuing evolutionary cycle.
\begin{equation}\label{reflection equation}
\texttt{prom} \leftarrow \texttt{prom} : \texttt{Reflection}(E_{\text{best}}, F_{\text{best}}).
\end{equation}

In our experiments, we generate automated policy feedback from environmental reward signals, as domain experts in relevant fields are not available. We acknowledge that obtaining feedback from human experts can be expensive. Nevertheless, it is important to note that within our framework, the number of rounds for feedback collection is specified by a predefined hyperparameter, which is typically kept low (in our experiments, it is set to 10). Therefore, in scenarios where human experts are accessible, incorporating their insights is feasible and can potentially enhance performance.

\section{Experiments}

\subsection{Experiment settings}\label{experiment settings}
For a comprehensive evaluation of \texttt{eSpark}'s capabilities, we perform detailed validations within two distinct industrial environments: the inventory management environment MABIM and the traffic signal control environment SUMO.
\begin{itemize}
    \item \textbf{MABIM setting}: MABIM simulates multi-echelon inventory management by modeling each stock-keeping unit (SKU) as an agent, mirroring real-world operations and profits within the MARL framework. The total reward is composed of multiple reward components. We utilize the total reward to identify the best-performing policy, while those components evaluate the policy's multifaceted performance to generate policy feedback. 
    We focus on three key challenges within inventory management: multiple echelons, capacity constraints and scalability, selecting corresponding scenarios for experiments.

    \item \textbf{SUMO setting}: SUMO is a traffic signal control environment in which each intersection is represented as an agent. It offers a variety of reward functions, and we use "the number of stopped vehicles" as the reward for evolutionary search, while other rewards are for policy feedback. The Average Delay, Average Trip Time, and Average Waiting Time metrics are chosen for evaluation~\cite{lu2023dualight}. We employ GESA~\cite{jiang2024general} to standardize intersections into 4-arm configurations. Each simulation spans 3600 seconds, with decisions at 15-second intervals.
    
    \item \textbf{Model setting}: We use IPPO as the base MARL framework for \texttt{eSpark} due to its DTDE structure, which is suitable for large-scale challenges. 
    But note that our approach can also be applied as a plugin in other MARL methods.
    We select GPT-4 for the $\texttt{LLM}_c$ and $\texttt{LLM}_g$ due to its superior comprehension and generation abilities. For each scenario, we conduct three runs with a batch size of $K=16$, training for $10$ iterations.
\end{itemize}
All training jobs are executed with an Intel(R) Xeon(R) Gold 6348 CPU and 4 NVIDIA RTX A6000 GPUs. In Appendix~\ref{environment details}, we provide a detailed introduction and setting for the environments and model. In Appendix~\ref{baseline details}, we give hyperparameter configurations and descriptions of each baseline method.

\subsection{Experiment results} \label{experiment results}
In this section, we present the key findings for \texttt{eSpark} in the MABIM and SUMO, highlighting the best and second best results in \textbf{bold} and \underline{underline}.
More detailed results are presented in Appendix~\ref{additional results}.
\subsubsection{Performance on MABIM}
Table \ref{100 challenge res} shows the results in the 100 SKUs scenarios in terms of capability constraints and multiple echelon challenges. With IPPO as the base MARL algorithm, \texttt{eSpark} not only outperforms IPPO in all scenarios but also exceeds the performance of all compared baselines in 4 out of 5 scenarios. For an in-depth analysis, we discuss the policy differences between IPPO and \texttt{eSpark} in Appendix~\ref{performance comparasion}, along with \texttt{eSpark}'s reflective mechanism and exploration function adjustments in Appendix~\ref{eSpark improvement}. While IPPO struggles to learn the intricate interplay among SKUs, \texttt{eSpark} excels particularly in navigating cooperation among SKUs and refining its search in a broad space, leading to marked improvements in managing capacity constraints and multi-echelon coordination.

\begin{table*}[htb]
\caption{Performance in MABIM, a higher profit indicates a better performance. The "Standard" scenario features a single echelon with sufficient capacity. The "2/3 echelons"  involves challenges of multi-echelon cooperation. The "Lower/Lowest" scenarios are the challenges where SKUs compete for insufficient capacity, while "500 SKUs scenarios" assess scalability. The `-' indicates CTDE algorithms are not researched in the scalability challenges.}\vspace{3mm}

\centering
\scalebox{0.74}{
\begin{tabular}{ccccccccccc}  
\toprule

\multirow{3}*{\textbf{Method}}&\multicolumn{10}{c}{\textbf{Avg. profits (K)}}\\ 
 & \multicolumn{5}{c}{100 SKUs scenarios} & \multicolumn{5}{c}{500 SKUs scenarios}\\ 
  & Standard & 2 echelons  & 3 echelons & Lower& \multicolumn{1}{c}{Lowest}  & \multicolumn{1}{c}{Standard} & 2 echelons & 3 echelons  & Lower  & Lowest \\ \midrule
IPPO & 690.6 & 1412.5 & 1502.9 & 431.1 & \underline{287.6} & 3021.2 & 7052.0 & 7945.7 & 3535.9 & \underline{2347.4} \\  
QTRAN & 529.6 & 1595.3 & 2012.2 & 70.1 & 19.5 & - & - & - & - & - \\  
QPLEX & 358.9 & 1580.7 & 704.2 & 379.8 & 259.3 & - & - & - & - & - \\  
MAPPO & 719.8 & 1513.8 & 1905.4 & 478.3 & 265.8 & - & - & - & - & - \\  \midrule
BS static & 563.7 & \underline{1666.6} & 2338.9 & 390.7 & -1757.5 & 3818.5 & 8151.2 & \underline{11926.3} & 3115.1 & -9063.8 \\  
BS dynamic & 684.2 & 1554.2 & \underline{2378.2} & \textbf{660.6} & -97.1 & 4015.7 & 8399.3 & 11611.1 & \textbf{3957.5} & 2008.6 \\  
$(S,s)$ & \underline{737.8} & 1660.8 & 1725.2 & 556.9 & 203.7 & \underline{4439.4} & \textbf{9952.1} & 10935.7 & 3769.3 & 2212.4 \\  \midrule

eSpark & \textbf{823.7} & \textbf{1811.4} & \textbf{2598.7} & \underline{579.5} & \textbf{405.0} & \textbf{4468.6} & \underline{9437.3} & \textbf{12134.2} & \underline{3775.7} & \textbf{2519.5} \\  \bottomrule
\end{tabular}  

}\vspace{-3mm}
\label{100 challenge res}
\end{table*}

In Table~\ref{100 challenge res}, we also present the performance outcomes of the \texttt{eSpark} algorithm in the scaling-up 500 SKUs scenarios. Due to the centralized nature of the CTDE methods, they struggle to scale to large-scale problems and therefore are not presented in the table. Despite IPPO's markedly inferior performance on scenarios when problems scale up, \texttt{eSpark} exhibits significant enhancements and consistently achieves optimal results across multiple scenarios. We attribute this improvement to \texttt{eSpark}'s action space pruning strategy, which effectively addresses the heightened exploration needs in scenarios with many agents, providing a clear advantage in such complex environments.

\subsubsection{Performance on SUMO}
To further assess \texttt{eSpark}'s capabilities across different tasks, we have compiled a summary of results in Table~\ref{sumo result} based on the SUMO environment. Similar to the outcomes in MABIM, \texttt{eSpark} consistently enhances the performance of the IPPO in all scenarios, and it has outperformed the CTDE baselines as well as domain-specific MARL baselines to achieve the best performance. Notably, even when IPPO alone is capable of good results (as seen in scenarios such as Grid 4$\times$4 and Cologne8), the pruning method designed in \texttt{eSpark} does not compromise the effectiveness of IPPO. We will delve further into the analysis of exploration functions produced by \texttt{eSpark} in Section~\ref{intelligent exploration functions}.

\begin{table}[ht]  
\vspace{-1mm}
\caption{Performance in SUMO, including the mean and standard deviation (in parentheses). A lower time indicates a better performance.}  
\centering  
\scalebox{0.57}{
    \begin{tabular}{ccccccccccc}  
    \toprule
    \multirow{2}*{\textbf{Method}} & \multicolumn{5}{c}{\textbf{Avg. delay (seconds)}} & \multicolumn{5}{c}{\textbf{Avg. trip time (seconds)}} \\ 
    & Grid 4$\times$4 & Arterial 4$\times$4 & Grid 5$\times$5 & Cologne8 & \multicolumn{1}{c}{Ingolstadt21} & Grid 4$\times$4 & Arterial 4$\times$4 & Grid 5$\times$5 & Cologne8 & Ingolstadt21 \\ \midrule
    FTC&  161.14 {\scriptsize (3.77)} & 1229.68 {\scriptsize (16.79)} & 820.88 {\scriptsize (24.36)} & 85.27 {\scriptsize (1.21)} & \underline{224.96 {\scriptsize (11.91)}} & 291.48 {\scriptsize (4.45)} & 676.77 {\scriptsize (13.70)} & 584.54 {\scriptsize (24.17)} & 145.93 {\scriptsize (0.84)} & \underline{352.06 {\scriptsize (9.29)}} \\
    MaxPressure & 63.39 {\scriptsize (1.34)} & 995.23 {\scriptsize (77.02)} & 242.85 {\scriptsize (16.04)} & 31.63 {\scriptsize (0.61)} & 228.64 {\scriptsize (15.83)} & 174.68 {\scriptsize (2.05)} & 702.09 {\scriptsize (23.61)} & 269.35 {\scriptsize (9.62)} & 95.67 {\scriptsize (0.62)} & 352.30 {\scriptsize (15.06)} \\ \midrule
    IPPO & \underline{48.40 {\scriptsize (0.45)}} & 884.73 {\scriptsize (38.94)} & 228.78 {\scriptsize (11.59)} & 27.60 {\scriptsize (1.70)} & 342.97 {\scriptsize (43.61)} & 160.12 {\scriptsize (0.60)} & 506.18 {\scriptsize (10.39)} & \underline{238.03 {\scriptsize (7.10)}} & \underline{91.41 {\scriptsize (1.60)}} & 464.50 {\scriptsize (43.30)}\\
    MAPPO & 51.25 {\scriptsize (0.58)} & 958.43 {\scriptsize (181.72)} & 958.43 {\scriptsize (181.72)} & 32.55 {\scriptsize (4.66)} & 347.59 {\scriptsize (47.59)} & \underline{160.01 {\scriptsize (0.54)}} & 757.40 {\scriptsize (242.00)} & 247.56 {\scriptsize (3.71)} & 94.31 {\scriptsize (1.77)} & 480.66 {\scriptsize (49.46)} \\
    CoLight & 53.40 {\scriptsize (1.89)} & \textbf{820.67 {\scriptsize (58.65)}} & 339.66 {\scriptsize (48.55)} & \underline{27.48 {\scriptsize (1.30)}} & 296.47 {\scriptsize (106.82)} & 165.77 {\scriptsize (1.89)} & \textbf{470.33 {\scriptsize (12.34)}} & 305.41 {\scriptsize (44.43)}  & 91.66 {\scriptsize (1.29)} & 410.59 {\scriptsize (97.29)} \\
    MPLight & 63.51 {\scriptsize (0.64)} & 1142.98 {\scriptsize (79.65)} & \underline{223.44 {\scriptsize (16.18)}} & 37.93 {\scriptsize (0.45)} & \textbf{196.74 {\scriptsize (9.88)}} & 172.47 {\scriptsize (0.89)} & 583.21 {\scriptsize (45.84)} & 255.49 {\scriptsize (6.26)} & 110.56 {\scriptsize (1.18)} & \textbf{331.42 {\scriptsize (11.79)}}\\
    \midrule
    eSpark & \textbf{48.36 {\scriptsize (0.32)}} & \underline{854.22 {\scriptsize (68.21)}} & \textbf{209.49 {\scriptsize (13.98)}} & \textbf{25.39 {\scriptsize (1.27)}} & 243.92 {\scriptsize (15.81)} & \textbf{159.74 {\scriptsize (0.44)}} & \underline{484.87 {\scriptsize (58.21)}} & \textbf{235.20 {\scriptsize (6.80)}} & \textbf{89.50 {\scriptsize (1.36)}} & 367.57 {\scriptsize (15.03)}\\
    \bottomrule
    \end{tabular}  
}
\label{sumo result}  
\vspace{-2mm}
\end{table}

\subsection{\texttt{eSpark} learns intelligent pruning methods} \label{intelligent exploration functions}
Given that \texttt{eSpark} employs the prior knowledge of the LLMs to craft its exploration function, our study aimed to investigate two critical aspects: (1) the validity of action space pruning via prior knowledge, and (2) the potential advantages of this method over rule-based heuristic pruning.

To address the questions raised, we devise two pruning strategies. First, we implement a \textbf{random pruning} method, wherein agents randomly exclude a portion of actions during decision-making to test the validity of knowledge-based pruning. Secondly, we utilize domain-related OR algorithms to implement \textbf{heuristic pruning} methods. For MABIM, actions are pruned using the $(S, s)$ policy and unbound limit, while for SUMO, pruning relies on MaxPressure to keep only a few actions with the highest pressure. The details of these methods are presented in Appendix~\ref{Heuristic Baselines}. Just like \texttt{eSpark}, these pruning strategies are integrated with IPPO during training but not execution. We conducted experiments under the same setting in Section~\ref{experiment settings}, with results presented in Tables~\ref{mabim or mask result} and Table~\ref{sumo or mask result}.

\begin{table}[htbp]  
\centering  
 \vspace{-2mm} 
\begin{minipage}[t]{0.45\linewidth}  
\centering  
\caption{Average performance changes on MABIM. All changes are relative to IPPO.}  
\scalebox{0.8}{
\begin{tabular}{ccc}  
\toprule
\multirow{2}*{\textbf{Method}}&\multicolumn{2}{c}{\textbf{Avg. profits change ratio (\%)}}\\ 
 & 100 SKUs & 500 SKUs \\ \midrule
Random pruning & 2.1 & -0.5  \\  
$(S,s)$ pruning & -25.9 & 15.5 \\  
Upbound pruning & -23.2 & -32.7 \\  
eSpark & \textbf{39.1} & \textbf{29.7}  \\  \bottomrule
\end{tabular}  
}
\label{mabim or mask result}
\end{minipage}  
\hspace{0.5cm}  
\begin{minipage}[t]{0.45\linewidth}  
\caption{Average performance changes on SUMO. All changes are relative to IPPO.}  
\centering  
\scalebox{0.8}{
\begin{tabular}{cccc}  
\toprule
\multirow{2}*{\textbf{Method}}&\multicolumn{3}{c}{\textbf{Avg. time change ratio (\%)}}\\ 
 & Delay & Trip time & Wait time\\ \midrule
Random pruning & -0.1 & 2.2 & -2.5 \\  
MaxPressure pruning & -0.5 & 1.5 & -0.1 \\  
eSpark & \textbf{-9.7} & \textbf{-5.7} & \textbf{-14.3}  \\  \bottomrule
\end{tabular}  
}
\label{sumo or mask result}
\end{minipage}  
\vspace{-2mm}  
\end{table}  

As shown in the tables, random pruning marginally affects performance by merely altering exploration rates without providing new insights. Heuristic pruning's impact varies with its design and context. In MABIM, $(S,s)$ pruning is less effective in the 100 SKUs scenario, as it restricts the already effective IPPO's exploration in smaller scales. However, it proves beneficial in the 500 SKUs scenario, where it guides the exploration and leads to better results. Upbound pruning consistently underperforms due to its overly simplistic heuristic. For SUMO, pressure-based pruning does not offer significant benefits. Nevertheless, \texttt{eSpark} demonstrates remarkable adaptability across all testing tasks, adeptly selecting pruning methods that substantially enhance results. Its knowledge-based generative technique and evolution capability enable it to master intelligent pruning strategies.

\begin{figure*}[h]
\centering
\vspace{-2mm}
\includegraphics[width=0.95\textwidth,height=0.155\textwidth]{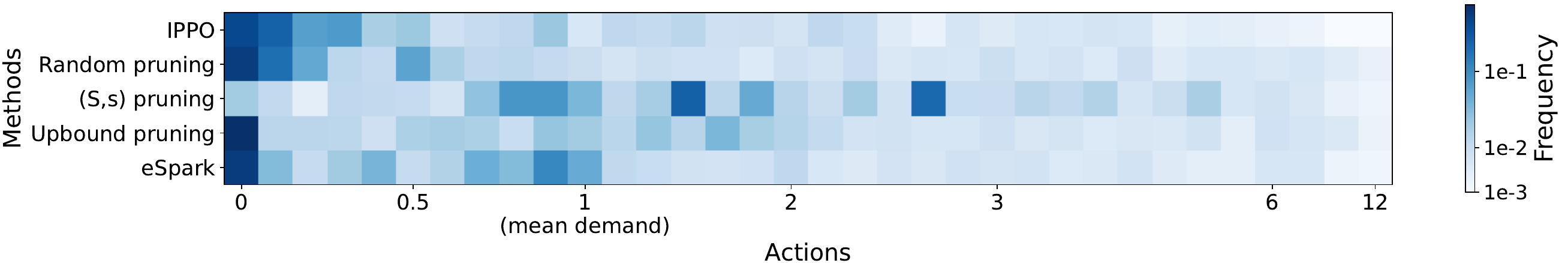}\vspace{-2mm}
\caption{Action selection frequency for IPPO and various pruning methods on the 100 SKUs Lowest scenario. "Actions" represents the replenishment quantity is a multiple of the mean demand within the sliding window. \texttt{eSpark} learns not only to minimize restocking but also to diversify with small purchases below the mean demand, balancing demand fulfillment and overflow prevention.}
\label{action freq heatmap}
\end{figure*}

Figure~\ref{action freq heatmap} presents a frequency heatmap of action selection for IPPO and various pruning methods in the 100 SKUs Lowest scenario. IPPO learns a minimally restocking strategy, risking unmet demand. Random pruning chooses actions more uniformly yet mirrors IPPO's pattern. $(S,s)$ pruning excessively exceeds mean demand, ignoring no-restock actions and leading to significant overflow. Upbound pruning typically avoids restocking, but prefers to purchase near the mean demand, which could result in overflow costs. In contrast, \texttt{eSpark} adopts a balanced policy, avoiding overstocking while diversifying its minor restocking strategies to meet demand without causing overflow.

\subsection{\texttt{eSpark} benefits from retention training and action space reduction}\label{retention training result} 

Extensive research has underscored the importance of reflection in LLM-driven content generation~\cite{ma2023eureka,nascimento2023gpt}. Herein, we focus on the effects of retention training and action pruning on \texttt{eSpark}'s performance.

We first design an ablation experiment, which we refer to as the {\textbf{eSpark w/o retention}}. The model parameters are initialized when an iteration is finished, and the newly generated exploration functions are equipped, after which the training starts from scratch. Given that the initialized model needs a more extensive number of steps to converge, we accordingly triple the training steps per iteration in comparison to the standard \texttt{eSpark}. Another ablation retains the retention training, while the only difference is that the LLMs and reflection are removed. We name this experiment \textbf{eSpark w/o LLM}. The comparative analysis of these two ablations is delineated in Table~\ref{100challengeresablation} and Table~\ref{sumo ablation}. 

\begin{table}[htbp]  
\centering  
  
\begin{minipage}[t]{0.45\linewidth}  
\centering  
\caption{Average performance change across 100 SKU scenarios in the MABIM environment. All changes are relative to IPPO.}  
\scalebox{0.8}{
\begin{tabular}{cc}  
\toprule
\multirow{1}*{\textbf{Method}}&\multicolumn{1}{c}{\textbf{Avg. profits change ratio (\%)}}\\ 
 \midrule
eSpark &  \textbf{39.1}  \\  
eSpark w/o retention& 24.0 \\  
eSpark w/o LLM &  -2.8 \\  \bottomrule
\end{tabular}  
}
\label{100challengeresablation}
\end{minipage}  
\hspace{0.5cm}  
\begin{minipage}[t]{0.45\linewidth}  
\caption{Average performance change in the SUMO environment. All changes are relative to IPPO. }  
\centering  
\scalebox{0.8}{
\begin{tabular}{cccc}  
\toprule
\multirow{2}*{\textbf{Method}}&\multicolumn{3}{c}{\textbf{Avg. time change ratio (\%)}}\\ 
 & Delay & Trip time & Wait time\\ \midrule
eSpark & \textbf{-9.7} & \textbf{-5.7} & \textbf{-14.3} \\  
eSpark w/o retention & -9.6 & -4.6 & -11.2 \\  
eSpark w/o LLM& -9.1 & -5.0 & -12.8 \\  \bottomrule
\end{tabular}  
}
\label{sumo ablation}
\end{minipage}  
  
\end{table}  
\vspace{-3mm}

The removal of retention training and LLMs both result in a decline in the performance of the \texttt{eSpark}. In the SUMO scenario, the performance gap between the two ablations and the complete \texttt{eSpark} is relatively small, whereas it is more pronounced in the MABIM scenarios. This can be attributed to the fact that MABIM involves a greater number of agents and a more complex observation space action space, where a superior pruning can significantly enhance the performance of MARL methods. Additionally, we observe that the lack of LLMs leads to a significant decrease in performance on MABIM, emphasizing the central role of knowledge-based action space pruning within the \texttt{eSpark}.

\section{Conclusions, limitations and future work}\label{conclusion}
We present \texttt{eSpark}, a novel framework for generating exploration functions, leveraging the advanced capabilities of LLMs to integrate prior knowledge, generate code and reflect, thereby refining the exploration in MARL. \texttt{eSpark} has surpassed its base MARL algorithm across all scenarios in both MABIM and SUMO environments. In terms of pruning strategies, pruning based on the prior knowledge from LLMs outshines both random and heuristic approaches. Ablation studies confirm the indispensable roles of retention training and action space reduction to \texttt{eSpark}'s success.

Nevertheless, \texttt{eSpark} also has certain limitations. First, currently \texttt{eSpark} is only applicable to tasks involving homogeneous agents. For heterogeneous agents, a potential method could be to generate distinct exploration functions for each agent; however, this approach becomes impractical when the number of agents is too large. Moreover, \texttt{eSpark} benefits from policy feedback to refine the exploration functions. When feedback is not informative regarding how to modify the exploration (e.g., in tasks with sparse rewards, end-of-episode feedback alone is too limited to develop automated feedback), \texttt{eSpark} may struggle to improve and need extra expert input for effective reflection.

Future work encompasses numerous potential directions. Existing research advocates for assigning different roles or categories to agents~\cite{pmlr-v139-christianos21a,wang2020roma}, which could offer a compromise for the application of \texttt{eSpark} in heterogeneous multi-agent systems. Furthermore, state-specific feedback for more granular improvement represents an intriguing avenue~\cite{subramanian2016exploration}. Our future endeavors will investigate these questions, striving to develop algorithms that are robust and exhibit strong generalizability.

\medskip

{
\small
\bibliographystyle{plain}  
\bibliography{example_paper}
\label{reference}


}

\newpage
\appendix

\section{Proofs} \label{appendix 1}
\setcounter{proposition}{0} 
\begin{proposition}
\begin{enumerate}
    \item[] 
    \item For any $E \in \{E\}$, $\{\bm{\pi}_E\}\subseteq\{\bm{\pi}\}$. If $E$ is non-trivial, then $\{\bm{\pi}_E\}\subset\{\bm{\pi}\}$.
    \item For any $\bm{\pi} \in \{\bm{\pi}\}$, there exists a non-trivial $E\in\{E\}$ such that $J(\bm{\pi}_{E}) \geq J(\bm{\pi})$.
\end{enumerate}
\end{proposition}
\begin{proof}
We begin by offering the proof of the first statement in Proposition~\ref{proposition 1}. We denote $\{\pi^k\}$ as the set of all possible policies for agent $k$, with each $\pi^k$ satisfying the following two conditions:
\begin{equation}
    \left\{\begin{matrix}
    \pi^k:\mathcal{O}^k\times\mathcal{A}^k\rightarrow [0,1] \\
    \sum_{a^k\in\mathcal{A}^k}\pi^k(a^k \mid  o^k)=1
    \end{matrix}\right.
\end{equation}
Let $\{\pi_E^k\}$ be the set of policies under an exploration function $E$. For every element in $\pi_E^k \in \{\pi_E^k\}$, one of the two situations exists:
\begin{enumerate}
    \item If $E$ is trivial, then $\pi_E^k(\cdot\mid o^k)=\pi^k(\cdot \mid  o^k)$, hence $\pi_E^k \in \{\pi^k\}$.
    \item If $E$ is non-trivial, then $\pi_E^k(\cdot\mid o^k)=\frac{\pi^k(\cdot \mid  o^k)\cdot E(o^k, \cdot)}{\sum_{a^k\in\mathcal{A}^k}\pi^k(a^k \mid  o^k)\cdot E(o^k, a^k)}$. It is clear that $ 0\leq\pi_E^k(\cdot\mid o^k)\leq 1$ and $\sum_{a^k\in\mathcal{A}^k}\pi_E^k(a^k \mid  o^k)=1$, thus $\pi_E^k \in \{\pi^k\}$.
\end{enumerate}

Therefore, we know that:
\begin{equation}\label{subset formula}
    \{\pi_E^k\} \subseteq \{\pi_k\}.
\end{equation}

Given that for $\forall \bm{\pi}(\cdot\mid s_t)=\prod_{k=1}^N\pi^k(\cdot\mid o^k)$, we have $\bm{\pi} \in \{\bm{\pi}\}$, where each $\pi^k$ belongs to $\{\pi^k\}$. According to Formula~\ref{subset formula}, it is known that for $\forall \boldsymbol{\pi}_E(\cdot\mid s_t) = \prod_{k=1}^N \pi_E^k(\cdot\mid o^k)$, it belongs to $\{\bm{\pi}\}$. Then:
\begin{equation}
    \{\boldsymbol{\pi}_E\} \subseteq \{\bm{\pi}\}.
\end{equation}
When $E$ is non-trivial, $\pi_E^k \in \{\pi^k\}$ still holds, but $\pi^k \in \{\pi_E^k\}$ may not be true (i.e., when $\pi^k(a^k\mid o^k)>0$ for $\forall a^k \in \mathcal{A}$, $\pi^k \notin \{\pi_E^k\}$). Hence we can get:
\begin{equation}
    \{\pi_E^k\} \subset \{\pi_k\}.
\end{equation}

In a similar manner, we can deduce that if there exists a $k\in [1,2,\dots, N]$ such that $\pi^k \notin \{\pi_E^k\}$, then $\boldsymbol{\pi} \notin \{\bm{\pi}_E\}$, which means:
\begin{equation}
    \{\boldsymbol{\pi}_E\} \subset \{\bm{\pi}\}.
\end{equation}
Therefore, we finish the proof of the first statement.

To proof the second statement, it is necessary to introduce a series of variables. We define the value function and state-action function for $\bm{\pi}$ as follows: $V_{\bm{\pi}}(s)=\mathbb{E}_{\textbf{a}_{0:\infty}\sim \bm{\pi}, s_{1:\infty}\sim P}
    \left[\sum_{t=0}^\infty {\rm \gamma^t r_t\mid s_0=s}\right]$ and $Q_{\bm{\pi}}(s,a)=\mathbb{E}_{\textbf{a}_{1:\infty}\sim \bm{\pi}, s_{1:\infty}\sim P}
    \left[\sum_{t=0}^\infty {\rm \gamma^t r_t\mid s_0=s, \textbf{a}_0=\textbf{a} }\right]$. 
    The advantage function is defined as $A_{\bm{\pi}}(s,\textbf{a})=Q_{\bm{\pi}}(s,\textbf{a})-V_{\bm{\pi}}(s)$. The joint exploration function is introduced as $\textbf{E}(s,\cdot)=\prod_{k=1}^N E(o^k,\cdot)$. The relationship between $V_{\bm{\pi}}(s)$ and $Q_{\bm{\pi}}(s,\textbf{a})$ can be formulated as:
\begin{equation}
    V_{\bm{\pi}}(s) = \sum_{\textbf{a}\in\boldsymbol{\mathcal{A}}^k}\bm{\pi}(\textbf{a} \mid s) Q_{\bm{\pi}}(s,\textbf{a})
\end{equation}

For a non-trivial $E$, the value function of $\bm{\pi}_E$ can be written as:
\begin{equation}
\begin{aligned}
    V_{\bm{\pi}_E}(s) &= \sum_{\textbf{a}\in\boldsymbol{\mathcal{A}}^k} \frac{\textbf{E}(s,\textbf{a})\bm{\pi}(\textbf{a} \mid s)}{\sum_{\textbf{a}\in\boldsymbol{\mathcal{A}}^k}\textbf{E}(s,\textbf{a})\bm{\pi}(\textbf{a}\mid s)}  Q_{\bm{\pi}}(s,\textbf{a})\\
    &= \frac{1}{\sum_{\textbf{a}\in\boldsymbol{\mathcal{A}}^k}\textbf{E}(s,\textbf{a})\bm{\pi}(\textbf{a}\mid s)} \sum_{\textbf{a}\in\boldsymbol{\mathcal{A}}^k} \textbf{E}(s,\textbf{a})\bm{\pi}(\textbf{a} \mid s) Q_{\bm{\pi}}(s,\textbf{a}) \\
    &= \frac{1}{\sum_{\textbf{a}\in\boldsymbol{\mathcal{A}}^k}\textbf{E}(s,\textbf{a})\bm{\pi}(\textbf{a}\mid s)} \left[V_{\bm{\pi}}(s) - \sum_{\textbf{a}\in\boldsymbol{\mathcal{A}}^k}\left(1-\textbf{E}(s, \textbf{a})\right)\bm{\pi}(\textbf{a} \mid s) Q_{\bm{\pi}}(s,\textbf{a}) \right].
\end{aligned}
\end{equation}

Thus, we have:
\begin{equation} \label{VE-V}
    \begin{aligned}
        V_{\bm{\pi}_E}(s) - V_{\bm{\pi}}(s) &= \frac{ V_{\bm{\pi}}(s) - \sum_{\textbf{a}\in\boldsymbol{\mathcal{A}}^k}\left(1-\textbf{E}(s, \textbf{a})\right)\bm{\pi}(\textbf{a} \mid s) Q_{\bm{\pi}}(s,\textbf{a})}{\sum_{\textbf{a}\in\boldsymbol{\mathcal{A}}^k}\textbf{E}(s,\textbf{a})\bm{\pi}(\textbf{a}\mid s)} - V_{\bm{\pi}}(s)\\
        &= \frac{ \left(1- \sum_{\textbf{a}\in\boldsymbol{\mathcal{A}}^k}\textbf{E}(s,\textbf{a})\bm{\pi}(\textbf{a}\mid s) \right) V_{\bm{\pi}}(s) -\sum_{\textbf{a}\in\boldsymbol{\mathcal{A}}^k}\left(1-\textbf{E}(s, \textbf{a})\right)\bm{\pi}(\textbf{a} \mid s) Q_{\bm{\pi}}(s,\textbf{a})}{\sum_{\textbf{a}\in\boldsymbol{\mathcal{A}}^k}\textbf{E}(s,\textbf{a})\bm{\pi}(\textbf{a}\mid s)}\\
        &=\frac{ \sum_{\textbf{a}\in\boldsymbol{\mathcal{A}}^k}\left(1-\textbf{E}(s, \textbf{a})\right)\bm{\pi}(\textbf{a} \mid s) V_{\bm{\pi}}(s) -\sum_{\textbf{a}\in\boldsymbol{\mathcal{A}}^k}\left(1-\textbf{E}(s, \textbf{a})\right)\bm{\pi}(\textbf{a} \mid s) Q_{\bm{\pi}}(s,\textbf{a})}{\sum_{\textbf{a}\in\boldsymbol{\mathcal{A}}^k}\textbf{E}(s,\textbf{a})\bm{\pi}(\textbf{a}\mid s)}\\
        &=-\frac{\sum_{\textbf{a}\in\boldsymbol{\mathcal{A}}^k}\left(1-\textbf{E}(s, \textbf{a})\right)\bm{\pi}(\textbf{a} \mid s) A_{\bm{\pi}}(s,\textbf{a})}{\sum_{\textbf{a}\in\boldsymbol{\mathcal{A}}^k}\textbf{E}(s,\textbf{a})\bm{\pi}(\textbf{a}\mid s)}.
    \end{aligned}
\end{equation}
When $-\sum_{\textbf{a}\in\boldsymbol{\mathcal{A}}^k}\left(1-\textbf{E}(s, \textbf{a})\right)\bm{\pi}(\textbf{a} \mid s) A_{\bm{\pi}}(s,\textbf{a}) \geq 0$, which means the expectation of the advantage value for the pruned actions is less than or equal to 0, then $V_{\bm{\pi}_E}(s) \geq V_{\bm{\pi}}(s)$. Because for every $s \in \mathcal{S}$, $\sum_{\textbf{a}\in\boldsymbol{\mathcal{A}}^k}A_{\bm{\pi}}(s,\textbf{a}) = \sum_{\textbf{a}\in\boldsymbol{\mathcal{A}}^k} Q_{\bm{\pi}}(s,\textbf{a})-V_{\bm{\pi}}(s) =0 $, there always exist actions for which the advantage function values are less than or equal to zero.

As $J(\bm{\pi}_{E})-J(\bm{\pi})=\mathbb{E}_{s_0\sim \rho^0}\left[V_{\bm{\pi}_E}(s_0) - V_{\bm{\pi}}(s_0) \right]$, if an exploration function $E$ can satisfy the condition that for all states $s\in \mathcal{S}$, the inequality $-\sum_{\textbf{a}\in\boldsymbol{\mathcal{A}}^k}\left(1-\textbf{E}(s, \textbf{a})\right)\bm{\pi}(\textbf{a} \mid s) A_{\bm{\pi}}(s,\textbf{a}) \geq 0$ holds, then it can be guaranteed that $J(\bm{\pi}_{E})\geq J(\bm{\pi})$. 

Therefore, we finish the proof of the second statement.
\end{proof}
\begin{proposition}
Let's consider a fully cooperative game with N agents, one state, and the joint action space $\{0,1\}^N$, where the reward is given by $r({\bm{\mathrm{0}}^0, \bm{\mathrm{1}}^N})=r_1$ and $r({\bm{\mathrm{0}}^{N-m},\bm{\mathrm{1}}^m})=-mr_2$ for all $m\neq N$, $r_1$, $r_2$ are positive real numbers. We suppose the initial policy is randomly and uniformly initialized, and the policy is optimized in the form of gradient descent. Let $p$ be the probability that the shared policy converges to the best policy, then:
    
\begin{equation}
    p=1-\sqrt[N-1]{\frac{r_2}{r_1+Nr_2}}.
\end{equation}
\end{proposition}

\begin{proof}
    Clearly, the best policy is the deterministic policy with joint action $({\bm{\mathrm{0}}^0, \bm{\mathrm{1}}^N)}$. 

Now, let the shared policy be $(1-\theta, \theta)$, where $\theta$ is the probability that an agent takes action 1. The expected reward can be written as:
\begin{equation}
\begin{aligned}
    J(\theta)&=\bm{\mathrm{Pr}}\left(\bm{a}^{1:N}=({\bm{\mathrm{0}}^0, \bm{\mathrm{1}}^N)}\right)\cdot r_1 - \sum_{t=0}^{N-1}\bm{\mathrm{Pr}}\left(\bm{a}^{1:N}=({\bm{\mathrm{0}}^{N-t}, \bm{\mathrm{1}}^t)}\right)\cdot t \cdot r_2\\
    &= \theta^N\cdot r_1 - \sum_{t=0}^{N-1} t\cdot C_N^t\theta^t(1-\theta)^{N-t} \cdot r_2\\
    &= \theta^N\cdot r_1 - \sum_{t=0}^{N} t\cdot C_N^t\theta^t(1-\theta)^{N-t}\cdot r_2 + N\cdot\theta^N\cdot r_2,
\end{aligned}
\label{initial J}
\end{equation}

where $C_N^t$ is a combinatorial number. We need to simplify $\sum_{t=0}^{N} t\cdot C_N^t\theta^t(1-\theta)^{N-t}$ for further analysis. Notice the structural similarity between the results and the binomial theorem:
\begin{equation}
    \left((1-\theta)+\theta\right)^N = \sum_{t=0}^NC_N^t\theta^t(1-\theta)^{N-t}.
    \label{binominal theroem}
\end{equation}
We take the derivative of $\theta$ on both sides of Formula~\ref{binominal theroem}. Because the left side is constant, its derivative is 0. Then:
\begin{equation}
\begin{aligned}
    0 &= \frac{d\sum_{t=0}^NC_N^t\theta^t(1-\theta)^{N-t}}{d\theta}\\
    &=\sum_{t=0}^NC_N^tt\cdot\theta^{t-1}\cdot(1-\theta)^{N-t}+C_N^t(N-t)\cdot (-1)\cdot(1-\theta)^{N-t-1}\cdot\theta^t\\
    &=\sum_{t=0}^NC_N^t(1-\theta)^{N-t-1}\theta^{t-1}\left((1-\theta)t-(N-t)\theta\right)\\
    &=\sum_{t=0}^NC_N^t(1-\theta)^{N-t-1}\theta^{t-1}(t-N\theta).
\end{aligned}
\end{equation}
Thus, we have:
\begin{equation}
    \begin{aligned}
        N\theta\sum_{t=0}^NC_N^t(1-\theta)^{N-t-1}\theta^{t-1}&=\sum_{t=0}^NtC_N^t(1-\theta)^{N-t-1}\theta^{t-1}\\
        N\theta\sum_{t=0}^NC_N^t(1-\theta)^{N-t}\theta^{t}&=\sum_{t=0}^NtC_N^t(1-\theta)^{N-t}\theta^{t}.
    \end{aligned}
\end{equation}
Notice that the left side of the equation is the expansion form of Formula~\ref{binominal theroem}, and the right side of the equation is the desired Formula, we can get:
\begin{equation}
    \sum_{t=0}^NtC_N^t(1-\theta)^{N-t}\theta^{t} = N\theta.
    \label{simple result}
\end{equation}

Bring Formula~\ref{simple result} back to Formula~\ref{initial J}, we get:
\begin{equation}
    J(\theta)=\theta^N\cdot r_1 -N\theta r_2 + N\cdot \theta^N\cdot r_2.
    \label{before deri}
\end{equation}
In order to maximise $J(\theta)$, we must maximise $\theta^N\cdot (r_1+Nr_2) -N\theta r_2$. Since the policy optimization usually adopts a gradient manner, we calculate the derivative of Formula~\ref{before deri} with respect to $\theta$ as:
\begin{equation}
    \frac{dJ(\theta)}{d\theta}=N\theta^{N-1}(r_1+Nr_2)-Nr_2.
\end{equation}

the point $\theta^*=\sqrt[N-1]{\frac{r_2}{r_1+Nr_2}}$ is the only zero of $\frac{dJ(\theta)}{d\theta}$. When $\theta\leq\theta^*$, $\frac{dJ(\theta)}{d\theta}\leq 0$; $\theta\geq\theta^*$, $\frac{dJ(\theta)}{d\theta}\geq 0$.

Remember we are trying to maximize $J(\theta)$ through a gradient way, and then the policy improves the parameters in the direction of the gradient. As the initial policy is randomly and uniformly initialized, the $\theta$ is uniformly distributed in the interval {\rm [0,1]}, then the probability that the shared policy converges to the best policy is:
\begin{equation}
    p=1-\sqrt[N-1]{\frac{r_2}{r_1+Nr_2}}.
\end{equation}
	Therefore, we finish the proof of Proposition~\ref{proposition 2}.
\end{proof}

\section{Pseudocode}\label{pseudocode sec}
In this Section, we give the pseudocode of our proposed \texttt{eSpark}. We denote the performance of policy $i$ as $G_i$, and the pseudocode of \texttt{eSpark} is shown in Algorithm~\ref{pseudocode}.
\begin{algorithm}[h]
\caption{\texttt{eSpark}} 
\begin{algorithmic}[1]  
\label{pseudocode}
\STATE \textbf{Input:} Initial prompt $\texttt{prom}$, LLM checker $\texttt{LLM}_c$, LLM code generator $\texttt{LLM}_g$, the evolution iteration number $N$, and sample batch size $K$

\STATE \textbf{Initialize:} policies $\phi_1^1, \phi_2^1, \dots, \phi_K^1$ 
\FOR{$i = 1$ \textbf{to} $N$}  
    \STATE // Exploration Function Generation
    \STATE $E_1, \dots, E_K \sim \texttt{LLM}_{c}(\texttt{prom},\texttt{LLM}_{g}(\texttt{prom}))$
    \STATE // Retention training
    \IF{$i \neq 1$}  
        \STATE $\phi_1^i, \phi_2^i, \dots, \phi_K^i \leftarrow \phi_{\text{best}}^{i-1}$   
    \ENDIF  
    \STATE // Evolutionary search
    \STATE $G_1, F_1 = \phi(E_1), \dots, G_K, F_K = \phi(E_K)$  
    \STATE // Reflection and Feedback
    \STATE $\text{best} \leftarrow \arg\max_k(G_1, G_2, \dots, G_K)$  
    \STATE $\texttt{prom} \leftarrow \texttt{prom}:\texttt{Reflection}(E_{\text{best}}, F_{\text{best}})$ 
\ENDFOR  
\STATE \textbf{Output:} $\phi_{best}^N$
\end{algorithmic}  
\end{algorithm}

\section{Detailed settings} \label{environment details} 
\subsection{MABIM details}
MABIM is a simulation environment dedicated to leveraging MARL to tackle the challenges inherent in inventory management problems. Within MABIM, each stock SKU at every echelon is represented as an autonomous agent. The decision-making process of each agent reflects the procurement requirements for the specific SKU at its corresponding echelon. 

\begin{figure}[h]
\centering
\includegraphics[width=0.65\textwidth]{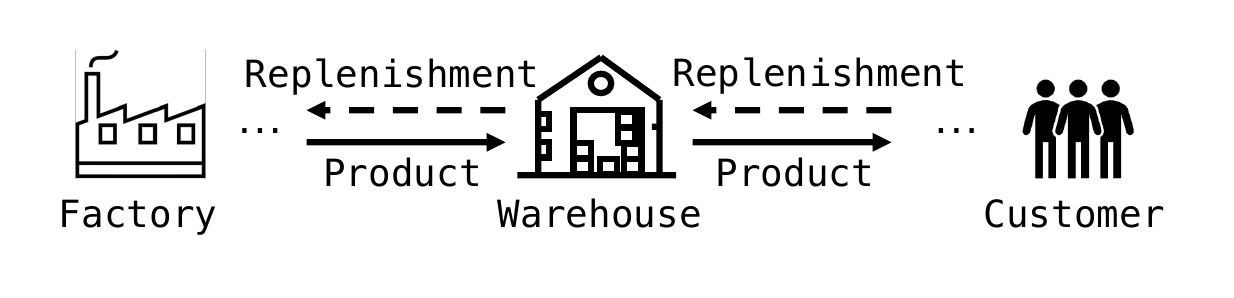}\vspace{-5mm}
\caption{MABIM inventory model.}
\end{figure}
Each time step involves the agent making decisions regarding replenishment quantities for SKUs and subsequently transitioning the environment to a new state. Let $M\in\mathbb{Z}^+$ be the total warehouses, with the first one being closest to customers, and $N\in\mathbb{Z}^+$ the total SKUs. Given a variable $X\in\{D,S,L\ldots\}$, $X_{i,j}^t$ represents its value for the $j$-th SKU in the $i$-th echelon at step $t$, with $0\le i<M$ and $0\le j<N$. Given the above notations, the main progression of a step can be described as follows:
\begin{equation*}
\begin{aligned}
    \centering
    & D_{i+1,j}^{t+1} = R_{i,j}^{t} & \text{(Replenish)} \\
    & S_{i,j}^t = \min(D_{i,j}^t, I_{i,j}^t) & \text{(Sell)}\\
    & A_{i,j}^{t} = \sum_{k=0}^{t-1} \mathbb{I}(k+L_{i,j}^k==t) \cdot S_{i+1,j}^t & \text{(Arrive)}\\ 
    & \gamma_i^t = \min\left(\frac{W_i - \sum_{j} I_{i,j}^t}{\sum_{j} A_{i,j}^t}, 1\right), B_{i,j}^t = \lfloor A_{i,j}^t \cdot \gamma_i^t \rfloor & \text{(Receive)} \\
    & I_{i,j}^{t+1} = I_{i,j}^t - S_{i,j}^t + B_{i,j}^t & \text{(Update)}\\
\end{aligned}
\end{equation*}
Here, $D, R, S, I, A, B \in \mathbb{Z}^+$ and $\mathbb{I}(\text{condition})$ is an indicator function that returns 1 if the condition is true, and 0 otherwise.
For the topmost echelon, orders are channeled to a super vendor capable of fulfilling all order demands at that level. Orders from other echelons are directed to their immediate upstream echelons, where the demands are satisfied based on the inventory levels of the upper echelons. The demand at the bottom echelon is derived from actual customer orders captured within real-world data sets. The reward function within MABIM is meticulously calibrated based on the economic realities of inventory management, integrating five fundamental elements: sales profit, order cost, holding cost, backlog cost, and excess cost. The summation of these elements constitutes the reward value, thereby incentivizing agents to optimize inventory control for enhanced profitability and operational efficacy.

MABIM incorporates challenges across five key categories: Scaling up, Cooperation, Competition, Generalization, and Robustness. We concentrate on the challenges associated with Scaling up, Cooperation, and Competition, as these challenges not only manifest in inventory management problems but also exist in a broad range of MARL tasks. We catalog the number of agents, challenges and degrees of difficulty within all the experimental scenarios in Table~\ref{challenge extent}. The specific setting of each scenario is given in Table~\ref{parameter in different challenge}:
\begin{table*}[h]
\caption{Tasks and corresponding challenges.  `\texttt{+}' denotes the extent of the challenges.}
\label{challenge extent}
\centering
\scalebox{1}{
\begin{tabular}{cccccccccc}
\toprule
 \multirow{2}*{Task name} &\multirow{2}*{Agents number}&\multicolumn{3}{c}{Challenge}\\ 
& & Scaling up& Cooperation & Competition & \\ \midrule
 Standard (100 SKUs)& 100& & &  \\
 2 echelons (100 SKUs)& 200& &\texttt{+} & \\
 3 echelons (100 SKUs)& 300& &\texttt{++} & \\
 Lower capacity(100 SKUs) &100  & & & \texttt{+}\\
 Lowest capacity (100 SKUs)&100 & & & \texttt{++}\\ \midrule
  Standard (500 SKUs)&500 &\texttt{+} & & \\
 2 echelons (500 SKUs)& 1000&\texttt{+} & \texttt{+} &\\
3 echelons (500 SKUs)& 1500&\texttt{+} & \texttt{++} &\\
 Lower capacity (500 SKUs)& 500&\texttt{+} & &\texttt{+} \\
 Lowest capacity (500 SKUs)& 500&\texttt{+} & & \texttt{++}\\
 
 \bottomrule\\
\end{tabular}}
\end{table*}

\begin{table}[ht]
    
    \caption{Experiments settings. "\#SKU * N" indicates N times the number of SKUs.}
    \label{parameter in different challenge}
    \centering
    \scalebox{1}{
        \begin{tabular}{ccccc}
            \toprule
            Task name & \#Echelon & \#SKU & Capacity per echelon  \\
            \midrule
            Standard (100 SKUs)& 1 & 100 & \#SKU*100 \\
             2 echelons (100 SKUs)& 2 & 100& \#SKU*100 & \\
 3 echelons (100 SKUs)& 3& 100&\#SKU*100\\
 Lower capacity(100 SKUs) &1 & 100&\#SKU*50  \\
 Lowest capacity (100 SKUs)&1 & 100& \#SKU*25 \\ 
  Standard (500 SKUs)&1 &500 & \#SKU*100\\
 2 echelons (500 SKUs)& 2&500 & \#SKU*100\\
3 echelons (500 SKUs)& 3&500 & \#SKU*100\\
 Lower capacity (500 SKUs)& 1&500 & \#SKU*50 \\
 Lowest capacity (500 SKUs)& 1&500&  \#SKU*25\\
            \bottomrule
        \end{tabular}
    }
\end{table}
For each scenario, we carry out three independent runs. Performance is reported as average test set profits from the top model in each run.
\subsection{SUMO details}
SUMO is an open-source, highly portable, microscopic and continuous road traffic simulation package designed to handle large road networks. In the SUMO simulation environment, each intersection is conceptualized as an autonomous agent equipped with an array of predefined traffic signal phases. These phases orchestrate the traffic flow across the intersection's multiple approaches. The selection of these phases, driven by the assessment of live traffic conditions, is aimed at attenuating road congestion and enhancing the fluidity of vehicular movement through the network, thus contributing to the overall efficiency of urban traffic management.

To conduct a thorough evaluation of each algorithm, we select a total of five scenarios from both synthetic and real-world datasets. These datasets encompass a diverse array of intersections, varying in number and type. The intersections are classified according to their configuration into three categories: bi-directional (2-arm), tri-directional (3-arm), and quadri-directional (4-arm), indicating the number of exit points at each junction. We summarize the type of each dataset, the number of intersections included, and the classification of these intersections in Table~\ref{sumo statistics}.

\begin{table}[ht]  
\centering  
\caption{The categories of each SUMO dataset, along with the number and types of intersections included.}  
\begin{tabular}{cccccc}  
\toprule
Dataset & Category  & Intersections number & 2-arm & 3-arm & 4-arm \\  
\midrule
Grid 4$\times$4& synthesis  & 16 & 0 & 0 & 16 \\  
Arterial 4$\times$4 & synthesis  & 16 & 0 & 0 & 16 \\  
Grid 5$\times$5~\cite{lu2023dualight} & synthesis  & 25 & 0 & 0 & 25 \\  
Cologne8 & real-world  & 8 & 1 & 3 & 4 \\  
Ingolstadt21 & real-world  & 21 & 0 & 17 & 4 \\  
\bottomrule
\end{tabular}  
\label{sumo statistics}  
\end{table}  

To facilitate a homogeneous observation and action space conducive to the deployment of various MARL algorithms, we employ the GEneral Scenario-Agnostic (GESA) framework to parse each intersection into a standardized 4-arm intersection with eight potential actions.
\subsection{Model details}
We employ IPPO as the base MARL algorithm for $\texttt{eSpark}$ due to its ability to scale to large-scale MARL challenges. We select GPT-4 as our $\texttt{LLM}_c$ and $\texttt{LLM}_g$, specifically opting for the $\texttt{2023-09-01-preview}$ version. The temperature of GPT-4 is set to 0.7, with no frequency penalty and presence penalty. For each scenario, we conduct three runs, setting the batch size for each generation of exploration functions to $K=16$. This batch size is chosen because it guarantees that the initial generation contains at least one executable exploration function for our environment. We limit the number of training iterations to $10$, as we observe that the performance for most scenarios tends to converge within this number of iterations.

\section{Baseline details} \label{baseline details}
In our experiments, we employed three categories of baselines: OR baselines, MARL baselines and pruning baselines. In Table~\ref{baseline details}, we list the characteristics and environment of all the baselines utilized in our study. In the rest of this section, we will elucidate the underlying principles of each OR baseline, articulate the design of the pruning baselines, and present the hyperparameter for the MARL baselines. 

\begin{table*}[hb]
\caption{All the baselines used in the experiments}
\centering
\scalebox{0.9}{
\begin{tabular}{cccccccccc}
\toprule
 \multirow{2}*{Algorithm name} & \multirow{2}*{OR baseline}& \multicolumn{2}{c}{MARL baseline} & \multirow{2}*{Pruning baseline}& \multicolumn{2}{c}{Used environment}\\ 
&   & CTDE & DTDE & & MABIM & SUMO \\ \midrule
Base stock (BS)~\cite{arrow1951optimal} &\checkmark & & &&\checkmark & \\
 $(S,s)$~\cite{blinder1990inventory}& \checkmark& & &&\checkmark &\\
 FTC~\cite{roess2004traffic} &\checkmark & & && & \checkmark\\
 MaxPressure~\cite{kouvelas2014maximum} &\checkmark & & && & \checkmark\\
IPPO& & &\checkmark &&\checkmark &\checkmark\\
QTRAN& &\checkmark & &&\checkmark&\\
QPLEX& &\checkmark & &&\checkmark &\\
MAPPO& &\checkmark & &&\checkmark& \checkmark\\
MPLight~\cite{chen2020toward}& & &\checkmark && &\checkmark \\
CoLight~\cite{wei2019colight}& &\checkmark & && & \checkmark\\
Ramdom pruning& & & & \checkmark & \checkmark &\checkmark\\
$(S,s)$ pruning & & & & \checkmark & \checkmark & \\
Upbound pruning & & & & \checkmark &\checkmark & \\
MaxPressure pruning& & & & \checkmark & & \checkmark\\
 \bottomrule\\
\end{tabular}}
\end{table*}
\subsection{OR baselines}
\subsubsection{Base stock algorithm}
The base stock algorithm constitutes a streamlined and efficacious approach for inventory control, whereby replenishment orders are initiated upon inventory below a predefined threshold level. This policy is traditionally acknowledged as a fundamental benchmark, favored for its straightforwardness and rapid implementation. The computation of the base stock level is facilitated through a programmatic methodology, as explicated in Equation~\ref{base stock policy equation}:

\begin{equation}
\label{base stock policy equation}
\begin{split}
{\bf max}\quad o_{i,j}^t &= \bar{p}_{i,j} \cdot S_{i,j}^t - \bar{c}_{i,j} \cdot S_{i+1,j}^t - \bar{h}_{i,j} \cdot I_{i,j}^{t+1} - \bar{c}_{i,j} \cdot T_{i,j}^0 - \bar{c}_{i,j} \cdot I_{i,j}^0 \\
{\bf s.t}\quad I_{i,j}^{t+1} &= I_{i,j}^t + S_{i+1,j}^{t-\bar{L}_{i,j}} - S_{i,j}^t \\
\quad T_{i,j}^{t+1} &= T_{i,j}^t - S_{i+1,j}^{t-\bar{L}_{i,j}} + S_{i+1,j}^t \\
S_{i,j}^t &= min(I_{i,j}^t, R_{i,j}^t) \\
T_{i,j}^0 &= \sum_{t=-\bar{L}_{i,j}}^{-1} S_{i+1,j}^t \\
z_{i,j} &= I_{i,j}^{t+1} + S_{i+1,j}^t + T_{i,j}^t \\
z_{i,j} &\in \mathbb{R^+}.\\
\end{split}
\end{equation}

In the above equations, $i$, $j$, and $t$ are indexes for the warehouse, SKU, and discrete time, respectively. The indicators $\bar{p}$, $\bar{c}$, $\bar{h}$, and $\bar{L}$ represent the average selling price, cost of procurement, cost of holding, and lead time. The variables $S$, $R$, $I$, and $T$ denote the quantities associated with sales, orders for replenishment, inventory in stock and inventory in transit. $o_{i,j}^t$ describes the profit objective, while $z_{i,j}$ is indicative of the base stock level.

We utilize two approaches for computing stock levels. The first approach, named \textbf{BS static}, involves calculating all base stock levels with historical data from the training set, which are then applied consistently to the test set. The levels remain unchanged during the test period. The second approach, termed as \textbf{BS dynamic}, computes stock levels directly on the test set relying on historical data and updates on a regular basis. 
\subsubsection{$(S,s)$ algorithm}
The $(S, s)$ inventory policy serves as a robust framework for managing stock levels. Under this policy, a restocking order is triggered once the inventory count falls below a predefined threshold, identified as $s$. The objective of this replenishment is to elevate the stock to its upper limit, designated as $S$. Empirical analyses have substantiated the efficacy of this protocol in optimizing inventory control processes. As a result, it is adopted as a benchmark heuristic, with the aim of algorithmically ascertaining the most efficacious $(S, s)$ parameter for each discrete SKU in the given inventory dataset. In our implementation, we conduct a search on the training set to identify the optimal values of $s$ and $S$, after which we apply these values consistently to the test set.
\subsubsection{Fixed-time control algorithm}
The Fixed-Time Control (FTC) algorithm is a traditional traffic signal control strategy predicated on predefined signal plans. These plans are typically designed based on historical traffic flow patterns and do not adapt to real-time traffic conditions. The FTC operates on a static schedule where the signal phases at intersections change at fixed intervals. This approach is straightforward and easy to implement but may not be optimal under variable traffic conditions due to its lack of responsiveness to dynamic traffic demands.

In our implementation, the FTC follows a fixed sequence of signal phases: 'WT-ET', 'NT-ST', 'WL-EL', 'NL-SL', 'WL-WT', 'EL-ET', 'SL-ST', 'NL-NT'. Here, 'W', 'E', 'N', and 'S' denote westbound, eastbound, northbound, and southbound traffic, respectively, while 'T' indicates through movement, and 'L' signifies a left turn. Each phase has a duration of 30 seconds

\subsubsection{MaxPressure algorithm}
The MaxPressure algorithm represents a more advanced traffic signal control strategy that dynamically adjusts signal phases in response to real-time traffic conditions. It calculates the "pressure" at each intersection, defined as the difference between the number of vehicles on the incoming and outgoing lanes. The algorithm aims to optimize traffic flow by selecting signal phases that reduce the maximum pressure across the network, thus alleviating congestion and enhancing network throughput. Unlike FTC, MaxPressure is adaptive and can continuously optimize signal timing based on the current traffic state, making it more suitable for managing fluctuating traffic volumes.

\subsection{Hyperparameters settings for MARL baselines}
The following table enumerates the hyperparameters employed during the training process for all MARL baselines. For all test scenarios, training is performed with a uniform suite of hyperparameters that have not undergone specialized fine-tuning. 

\begin{table}[t]
    \caption{Hyperparameters of MARL Algorithms Used in MABIM and SUMO Environments. `-' indicates that the algorithm is not set or does not contain this hyperparameter.}  
    \label{hyperparameters for MARL}  
    \centering  
    \scalebox{0.85}{
    \begin{tabular}{ccccc|ccc}  
        \toprule  
        \multirow{2}*{Hyperparameter} & \multicolumn{4}{c}{MABIM environment} & \multicolumn{3}{c}{SUMO environment}\\
         & IPPO & QTRAN & QPLEX & \multicolumn{1}{c}{MAPPO} & IPPO & CoLight & MPLight\\  
        \midrule  
        Training steps & 5020000 & 5020000 &5020000 & 5020000 &2400000 & 2400000 &2400000 \\
        Discount rate & 0.985 & 0.985& 0.985 & 0.985& 0.985 & 0.9& 0.9  \\  
        Optimizer & Adam & Adam& Adam & Adam& Adam & RMSProp& RMSProp \\  
        Optimizer alpha & 0.99 & 0.99 & 0.99 & 0.99  & 0.99 & 0.95 & 0.95\\  
        Optimizer eps & 1e-5 & 1e-5 & 1e-5 & 1e-5 & 1e-5 & 1e-7 & 1e-7 \\  
        Learning rate & 5e-4 & 5e-4 & 5e-4 & 5e-4& 5e-4 & 1e-3 & 1e-3 \\  
        Grad norm clip & 10 & 10& 10 & 10  & 10 & -& -  \\  
        Eps clip & 0.2 & -&-&0.2  & 0.2 & -&- \\  
        Critic coef & 0.5 & - &-&0.5 & 0.5 & - &-\\  
        Entropy coef & 0 & - &-&0 & 0 & - &-\\  
        Accumulated episodes &  4 & 8&  8 & 4 &  4 & 10 &  10 \\
        Number of neighbors&-&-&-&-& - & 5 & -\\
        \bottomrule  
    \end{tabular}  
    }
\end{table}  

\subsection{Pruning baselines} \label{Heuristic Baselines}
\subsubsection{Random pruning}
Random pruning is implemented by randomly masking a certain percentage of the available actions. During the action selection process, each agent will have $p$ percent of its available actions randomly masked. To balance the observability of the pruning's impact with the preservation of the algorithm's capacity to utilize prior experience, we set $p=0.3$.

\subsubsection{$(S,s)$ pruning}
According to the $(S,s)$ algorithm, for a given SKU, a replenishment quantity of $\Delta = S - s$ is ordered when the current inventory level falls below the threshold $s$; otherwise, no order is placed. We extend the reference replenishment quantity $\Delta$ to a range $[r_1\times\Delta, r_2\times\Delta]$, where $r_1, r_2$ are both real numbers and $0 \leq r_1 \leq 1$ and $r_2 \geq 1$. Actions within this interval are deemed available, while those outside of this range are masked. In our implementation, we select $r_1=0.5$ and $r_2=2$.

\subsubsection{MaxPressure pruning}
The MaxPressure pruning method utilizes the heuristic concept of "pressure" at an intersection to prune actions. We calculate the pressure associated with each action, and these pressures are then ranked. The actions with the top-k highest pressures are rendered available for selection. Actions not meeting this threshold are subsequently masked.

A standardized intersection warped through the GESA is modeled as a four-arm intersection comprising eight potential actions. We empirically set $k=4$ to ensure effective pruning while maintaining a sufficient number of available actions.
\section{Additional results}\label{additional results}
\subsection{Additional main results}
In this section, we present the complete main experimental results for SUMO in Section~\ref{experiment results}, and give the consumption of GPT tokens here.
\begin{table}[ht]  \label{additional sumo result}
\centering  
\caption{Detailed performance in SUMO, includes the mean and standard deviation.}  
\scalebox{0.85}{
\begin{tabular}{ccccccc}  
\toprule
Method & Metric & Grid 4$\times$4 & Arterial 4$\times$4 & Grid 5$\times$5 & Cologne8 & Ingolstadt21 \\  
\midrule
\multirow{3}*{FTC}& Delay & 161.14$\pm$ 3.77 & 1229.68$\pm$16.79 & 820.88$\pm$24.36 & 85.27$\pm$1.21 & 224.96$\pm$11.91 \\  
& Trip time & 291.48$\pm$4.45 & 676.77$\pm$13.70 & 584.54$\pm$24.17 & 145.93$\pm$0.84 & 352.06$\pm$9.29 \\  
& Wait time & 155.66$\pm$3.42 & 521.86$\pm$13.33 & 441.63$\pm$21.13 & 58.92$\pm$0.68 & 161.22$\pm$7.88 \\  
\midrule  
\multirow{3}*{MaxPressure} & Delay & 63.39$\pm$1.34 & 995.23$\pm$77.02 & 242.85$\pm$16.04 & 31.63 $\pm$ 0.61 & 228.64$\pm$15.83 \\  
& Trip time & 174.68$\pm$2.05 & 702.09$\pm$23.61 & 269.35$\pm$9.62 & 95.67$\pm$0.62 & 352.30 $\pm$ 15.06 \\  
& Wait time & 37.37$\pm$1.06 & 511.06$\pm$22.55 & 114.36$\pm$6.48 & 11.03$\pm$ 0.28 & 159.44$\pm$13.34 \\  
\midrule  
\multirow{3}*{IPPO} & Delay & 48.40$\pm$0.45 & 884.73$\pm$38.94 & 228.78$\pm$11.59 & 27.60$\pm$1.70 & 342.97$\pm$43.61 \\  
& Trip time & 160.12$\pm$0.60 & 506.18$\pm$10.39 & 238.03$\pm$7.10 & 91.41$\pm$1.60 & 464.50$\pm$43.30 \\  
& Wait time & 22.69$\pm$0.38 & 435.44$\pm$77.54 & 91.84$\pm$6.31 & 7.70$\pm$0.82 & 267.51$\pm$40.53 \\  
\midrule  
\multirow{3}*{MAPPO} & Delay & 51.25$\pm$0.58 & 958.43$\pm$181.72 & 221.62$\pm$20.73 & 32.55$\pm$4.66 & 347.59$\pm$47.59 \\  
& Trip time & 160.01$\pm$0.54 & 757.40$\pm$242.00 & 247.56$\pm$3.71 & 94.31$\pm$1.77 & 480.66$\pm$49.46 \\  
& Wait time & 25.41$\pm$0.54 & 609.80$\pm$255.22 & 97.10$\pm$5.22 & 9.39$\pm$1.53 & 283.59$\pm$43.20 \\  
\midrule  
\multirow{3}*{MPLight} & Delay & 63.51$\pm$0.64 & 1142.98$\pm$79.65 & 223.44$\pm$16.18 & 37.93$\pm$0.45 & 196.74$\pm$9.88 \\  
& Trip time & 172.47$\pm$0.89 & 583.21$\pm$45.84 & 255.49$\pm$6.26 & 110.56$\pm$1.18 & 331.42$\pm$11.79 \\  
& Wait time & 40.32$\pm$0.96 & 366.27$\pm$58.03 & 126.42$\pm$5.31 & 12.98$\pm$0.57 & 126.09$\pm$13.60 \\  
\midrule  
\multirow{3}*{CoLight} & Delay & 53.40$\pm$1.89 & 820.67$\pm$58.65 & 339.66$\pm$48.55 & 27.48$\pm$1.30 & 296.47$\pm$106.82 \\  
& Trip time & 165.77$\pm$1.89 & 470.33$\pm$12.34 & 305.41$\pm$44.43 & 91.66$\pm$1.29 & 410.59$\pm$97.29 \\  
& Wait time & 27.25$\pm$1.64 & 312.47$\pm$16.63 & 157.65$\pm$35.69 & 9.35$\pm$1.09 & 215.98$\pm$90.62 \\  
\midrule  
\multirow{3}*{eSpark} & Delay & 48.36$\pm$0.32 & 854.22$\pm$68.21 & 209.49 $\pm$13.98 & 25.39$\pm$1.27 & 243.92$\pm$15.81 \\  
& Trip time & 159.74$\pm$0.44 & 484.87$\pm$58.21 & 235.20$\pm$6.80 & 89.50$\pm$1.36 & 367.57$\pm$15.03 \\  
& Wait time & 22.58$\pm$0.29 & 328.82$\pm$61.70 & 88.38$\pm$4.41 & 6.94$\pm$0.38 & 180.09$\pm$13.84 \\  
\bottomrule
\end{tabular}  }
\label{your_label_here}  
\end{table}  

Since we do not design specific prompts for different scenario tasks within the same environment, we calculate the average token consumption for all scenarios with each environment and display it in Table~\ref{token assumption}.

\begin{table}[ht]
    \label{token assumption}
    \caption{Average token assumption for MABIM and SUMO.}
    \label{parameter in different challenge}
    \centering
    \scalebox{1}{
        \begin{tabular}{cc} 
            \toprule
            Environment & Token assumption (M)  \\
            \midrule
            MABIM&  3.0 \\
             SUMO& 2.6  \\
            \bottomrule
        \end{tabular}
    }
\end{table}

\subsection{Detailed results of the pruning methods}
In this section, we provide the detailed results for multiple pruning baselines as discussed in Section~\ref{intelligent exploration functions}, along with results of \texttt{eSpark} for comparison.
\begin{table*}[htb]
\vspace{-3mm}
\caption{Detailed performance of various pruning methods in MABIM.}\vspace{3mm}%
\centering
\scalebox{0.71}{
\begin{tabular}{ccccccccccc}  
\toprule
\multirow{3}*{\textbf{Method}}&\multicolumn{10}{c}{\textbf{Avg. profits (K)}}\\ 
 & \multicolumn{5}{c}{100 SKUs scenarios} & \multicolumn{5}{c}{500 SKUs scenarios}\\ 
  & Standard & 2 echelons  & 3 echelons & Lower& \multicolumn{1}{c}{Lowest}  & \multicolumn{1}{c}{Standard} & 2 echelons & 3 echelons  & Lower  & Lowest \\ \midrule
Random pruning & 733.0 & 1407.6 & 1426.6 & 511.6 & 262.0 & 2718.0 & 8667.4 &9464.1 & 2535.5 & 2202.1 \\  
$(S,s)$ pruning & 394.4 & 832.3 & 933.4 & 441.1 & 258.3 & 3884.3 & 9248.3 & 10282.1 & 3517.9 & 2085.6 \\  
Upbound pruning & 745.0 & 630.2 & -2.2 & 557.7 & 294.7 & 3261.3 & 2473.6 & 1657.6 & 2833.0 & 2167.7 \\  
eSpark & 823.7 & 1811.4 & 2598.7 & 579.5 & 405.0 & 4468.6 & 9437.3 & 12134.2 & 3775.7 & 2519.5 \\ \bottomrule
\end{tabular}  
}
\label{pruning methods raw data mabim}
\end{table*}

\begin{table}[ht]  
\centering  
\caption{Detailed performance of various pruning methods in SUMO, includes the mean and standard deviation.}  
\scalebox{0.8}{
\begin{tabular}{ccccccc}  
\toprule
Method & Metric & Grid 4$\times$4 & Arterial 4$\times$4 & Grid 5$\times$5 & Cologne8 & Ingolstadt21 \\  
\midrule
\multirow{3}*{Ramdom pruning}& Delay & 49.07$\pm$0.36 & 858.33$\pm$48.20 & 238.57$\pm$9.45 & 25.89$\pm$1.34 & 353.38$\pm$24.39 \\  
& Trip time & 160.13$\pm$0.58 & 548.08$\pm$61.84 & 241.92$\pm$9.60 & 89.75$\pm$1.26 & 478.53$\pm$22.54 \\  
& Wait time& 22.66$\pm$0.14 & 387.25$\pm$43.93 & 93.49$\pm$8.09 & 7.03$\pm$0.37 & 281.97$\pm$21.90 \\  
\midrule  
\multirow{3}*{MaxPressure pruning} & Delay & 48.78$\pm$0.37 & 890.04$\pm$121.50 & 234.27$\pm$14.28 & 26.26$\pm$0.36 & 337.02$\pm$62.18 \\  
& Trip time & 160.72$\pm$0.17 & 533.36$\pm$78.08 & 253.68$\pm$17.68 & 90.36$\pm$0.88 & 448.11$\pm$65.32 \\  
& Wait time& 23.03$\pm$0.56 & 391.96$\pm$78.34 & 102.29$\pm$10.34 & 7.33$\pm$0.12 & 257.98$\pm$61.91 \\  
\midrule  
\multirow{3}*{eSpark} & Delay & 48.36$\pm$0.32 & 854.22$\pm$68.21 & 209.49 $\pm$13.98 & 25.39$\pm$1.27 & 243.92$\pm$15.81 \\  
& Trip time & 159.74$\pm$0.44 & 484.87$\pm$58.21 & 235.20$\pm$6.80 & 89.50$\pm$1.36 & 367.57$\pm$15.03 \\  
& Wait time& 22.58$\pm$0.29 & 328.82$\pm$61.70 & 88.38$\pm$4.41 & 6.94$\pm$0.38 & 180.09$\pm$13.84 \\  
\bottomrule
\end{tabular}  }
\label{pruning raw data sumo}  
\end{table}  

\subsection{Detailed results of the ablations}
In this section, we provide the detailed results for ablations in Section~\ref{retention training result}, along with results of \texttt{eSpark} for comparison.
\begin{table*}[htb]
\vspace{-3mm}
\caption{Detailed performance of ablations in MABIM.}\vspace{3mm}%
\centering
\scalebox{0.8}{
\begin{tabular}{cccccc}  
\toprule
\multirow{2}*{\textbf{Method}}&\multicolumn{5}{c}{\textbf{Avg. profits (K)}}\\ 
  & Standard & 2 echelons  & 3 echelons & Lower& \multicolumn{1}{c}{Lowest} \\ \midrule
eSpark w/o retention & 719.0 & 1806.1 & 2388.6 & 547.7 & 294.1 \\  
eSpark w/o LLM & 754.7 & 1538.9 & 1109.9 & 536.7 & 198.5 \\  
eSpark & 823.7 & 1811.4 & 2598.7 & 579.5 & 405.0 \\ \bottomrule
\end{tabular}  
}
\label{ablation raw data mabim}
\end{table*}\vspace{-5mm}

\begin{table}[ht]  
\centering  
\caption{Detailed performance of ablations in SUMO, includes the mean and standard deviation.}  
\scalebox{0.8}{
\begin{tabular}{ccccccc}  
\toprule
Method & Metric & Grid 4$\times$4 & Arterial 4$\times$4 & Grid 5$\times$5 & Cologne8 & Ingolstadt21 \\  
\midrule
\multirow{3}*{eSpark w/o retention}& Delay & 48.75$\pm$0.49 & 851.56$\pm$37.98 & 211.07$\pm$22.01 & 25.18$\pm$0.51 & 246.05$\pm$14.88 \\  
& Trip time & 160.42$\pm$0.54 & 487.05$\pm$65.66 & 248.96$\pm$15.04 & 89.23$\pm$0.49 & 363.28$\pm$14.83 \\  
& Wait time& 22.97$\pm$0.39 & 338.64$\pm$67.07 & 97.77$\pm$10.25 & 7.14$\pm$0.15 & 175.79$\pm$12.47 \\  
\midrule  
\multirow{3}*{eSpark w/o LLM} & Delay & 48.67$\pm$0.48 & 854.10$\pm$63.38 & 212.25$\pm$16.13 & 24.70$\pm$0.56 & 257.14$\pm$40.50 \\  
& Trip time & 159.90$\pm$0.56 & 491.49$\pm$67.52 & 238.33$\pm$9.41 & 88.57$\pm$0.60 & 376.46$\pm$40.01 \\  
& Wait time& 22.99$\pm$0.43 & 342.62$\pm$73.88 & 90.89$\pm$6.60 & 6.69$\pm$0.23 & 188.14$\pm$37.03 \\  
\midrule  
\multirow{3}*{eSpark} & Delay & 48.36$\pm$0.32 & 854.22$\pm$68.21 & 209.49 $\pm$13.98 & 25.39$\pm$1.27 & 243.92$\pm$15.81 \\  
& Trip time & 159.74$\pm$0.44 & 484.87$\pm$58.21 & 235.20$\pm$6.80 & 89.50$\pm$1.36 & 367.57$\pm$15.03 \\  
& Wait time& 22.58$\pm$0.29 & 328.82$\pm$61.70 & 88.38$\pm$4.41 & 6.94$\pm$0.38 & 180.09$\pm$13.84 \\  
\bottomrule
\end{tabular}  }
\label{ablation raw data sumo}  
\end{table}  

\section{Policy performance analysis}\label{performance comparasion}
To gain a deeper understanding of the policy difference between \texttt{eSpark} and IPPO, we select the capacity limit and multiple echelons challenges within the 100 SKUs scenario as representative cases, presenting in Figure~\ref{coorporation_and_competition} the daily profit of \texttt{eSpark} and IPPO on the test dataset challenged with capacity limit and multiple echelons. In the capacity limit challenges, a high daily overflow ratio and low fulfillment ratio suggest that IPPO falls short in mastering the adjustment of restocking quantities for individual agents when capacity is limited, leading to overstocking and substantial overflow. Concurrently, this prevents SKUs required by consumers from being accommodated, culminating in an exceedingly low fulfillment ratio. In multiple echelon challenges, the fulfillment ratio at each echelon gradually decreases over time, indicating that IPPO struggles to comprehend and learn the intricate interplay required for cooperation among various echelons, thereby inadequately fulfilling the demands of each echelon. Such shortcomings not only diminish potential profits but also subject the system to considerable backlog expenses. However, through action space pruning, evolutionary search, and reflection, \texttt{eSpark} manages to reduce the search within the vast space, selecting the most effective exploration functions and continuously improving. This approach significantly reduces overflow in the capacity limit scenario and successfully learns suitable cooperation methods for multiple echelons.

\begin{figure}[h]
\centering
\includegraphics[width=0.95\textwidth]{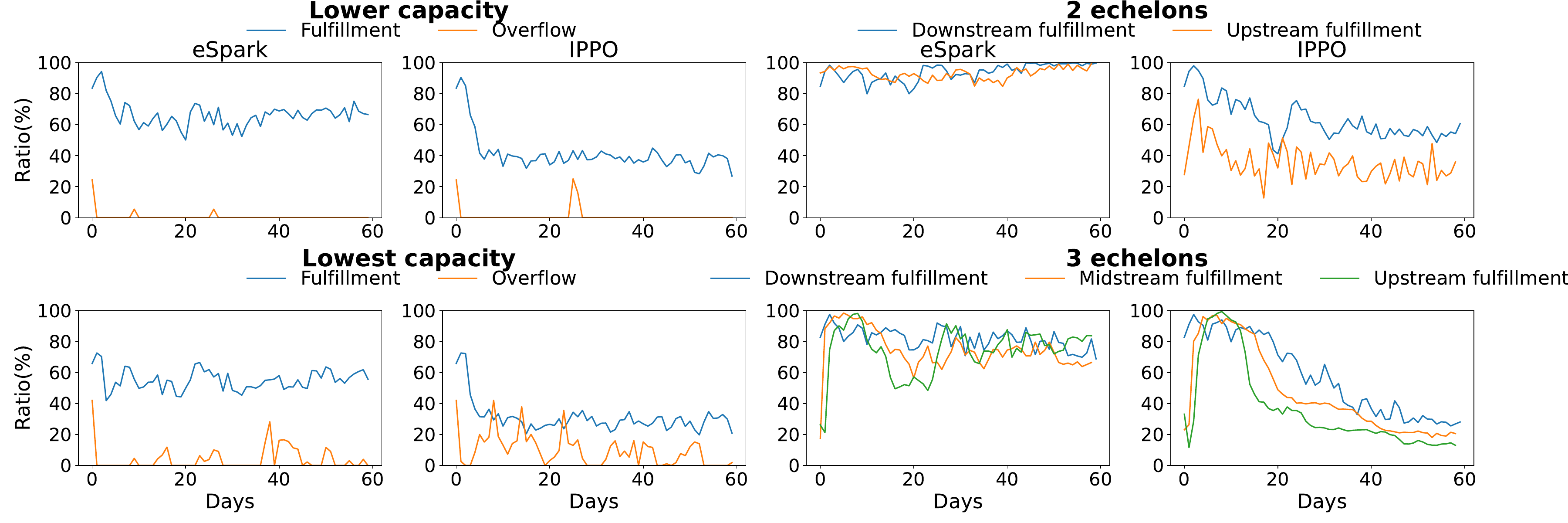}\vspace{-5mm}
\caption{The performance comparison between \texttt{eSpark} and IPPO in 100 SKUs scenarios. In capacity-limited scenarios, \texttt{eSpark} strives to meet the demands while minimizing overflow costs, boasting a lower overflow ratio and a higher fulfillment ratio. In the multiple-echelon challenge, \texttt{eSpark} achieves nuanced collaboration across different echelons, ensuring high fulfillment ratios.} \vspace{-2mm}
\label{coorporation_and_competition}
\end{figure}
\newpage
\section{Full prompts} \label{full prompts}
We reference the prompt design outlined in the Eureka~\cite{ma2023eureka} and adapt it specifically for exploration function generation. Our prompt provides general guidance on the design of exploration functions, specific code formatting suggestions, feedback, and recommendations for improvement. We present our prompts for MABIM below.

\begin{figure}[H]
\centering
\includegraphics[width=1\textwidth]{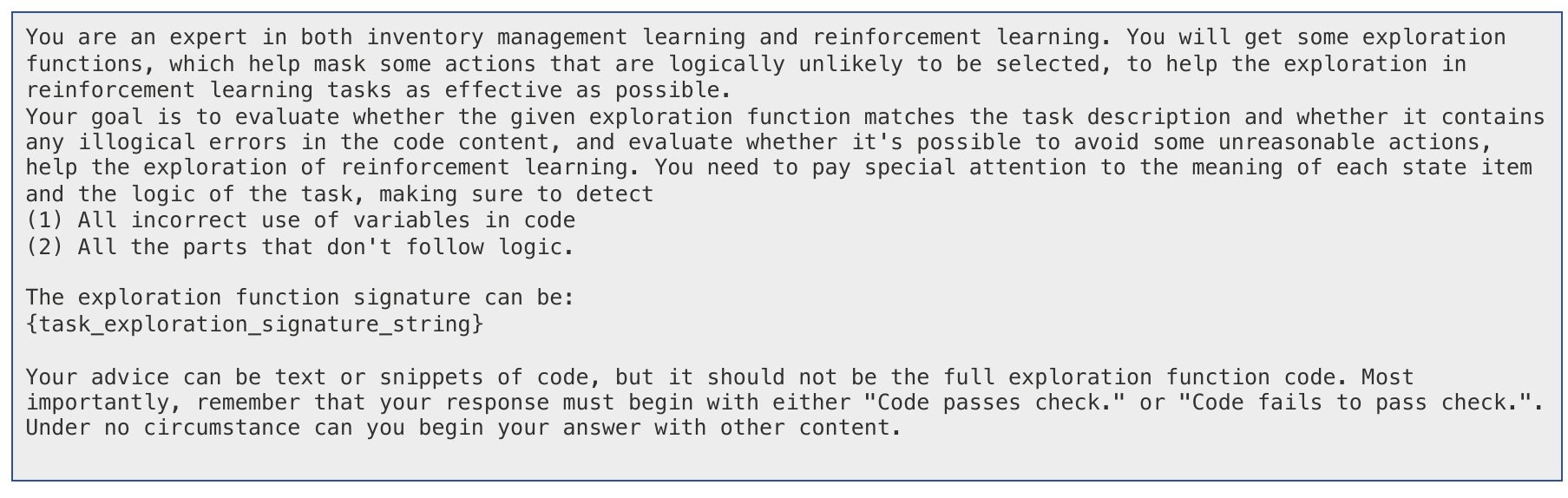}\vspace{-5mm}
\caption{System prompt for $\texttt{LLM}_c.$} \vspace{-3mm}
\end{figure}

\begin{figure}[H]
\centering
\includegraphics[width=1\textwidth]{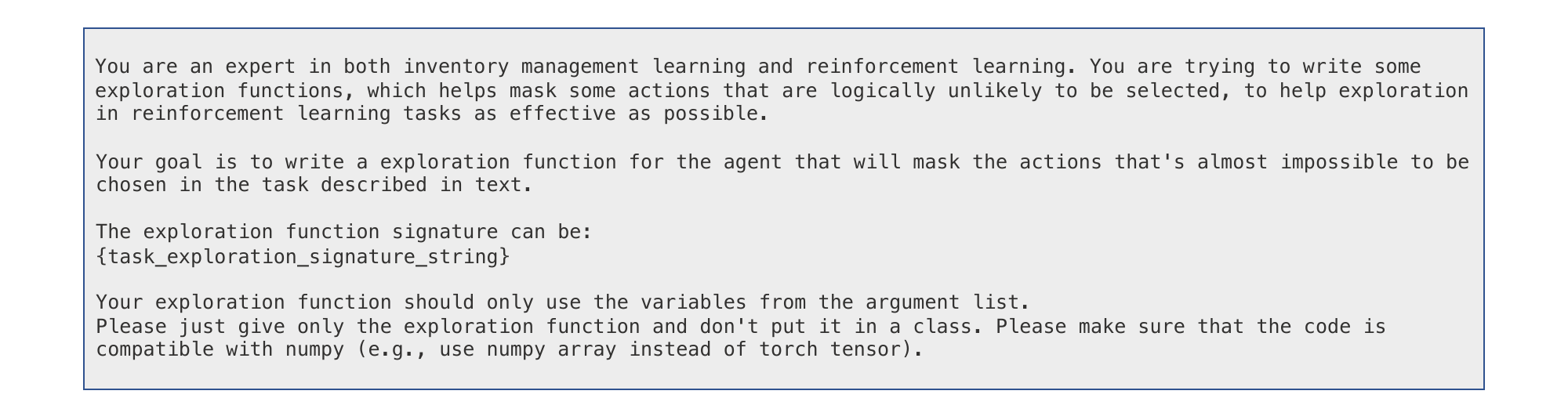}\vspace{-5mm}
\caption{System prompt for $\texttt{LLM}_g$.} 
\end{figure}

\begin{figure}[H]
\centering
\includegraphics[width=1\textwidth]{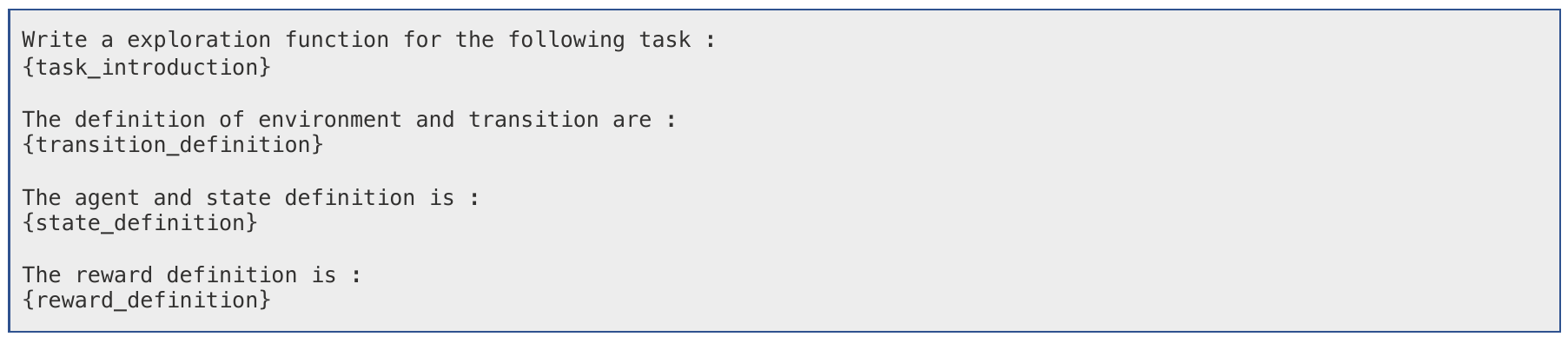}\vspace{-5mm}
\caption{Initial prompt for $\texttt{LLM}_g$.} 
\end{figure}

\begin{figure}[H]
\centering
\includegraphics[width=1\textwidth]{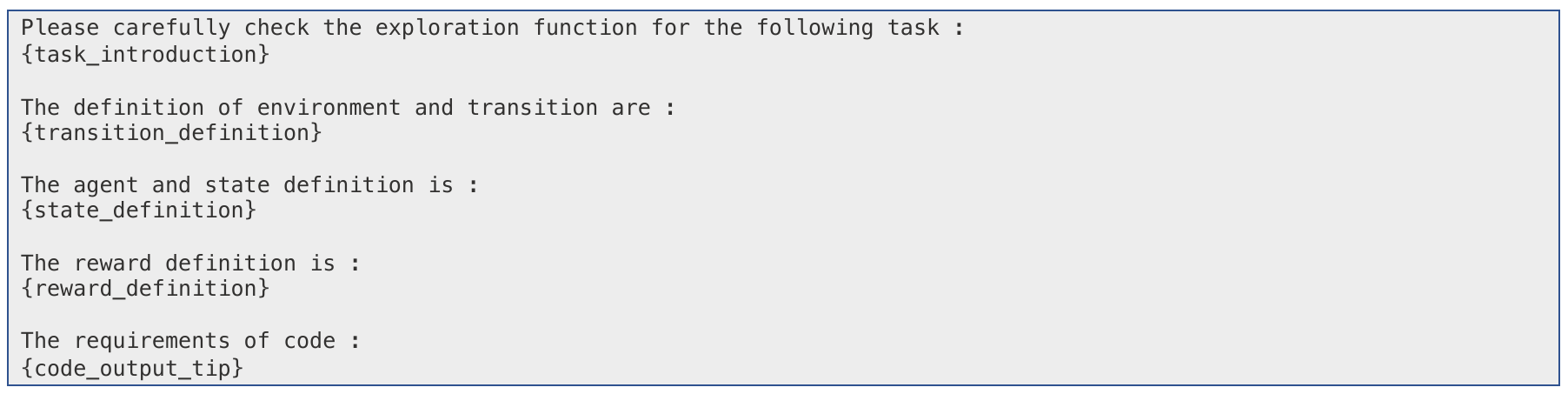}\vspace{-4mm}
\caption{Initial prompt for $\texttt{LLM}_c$.} 
\end{figure}

\begin{figure}[H]
\centering
\includegraphics[width=1\textwidth]{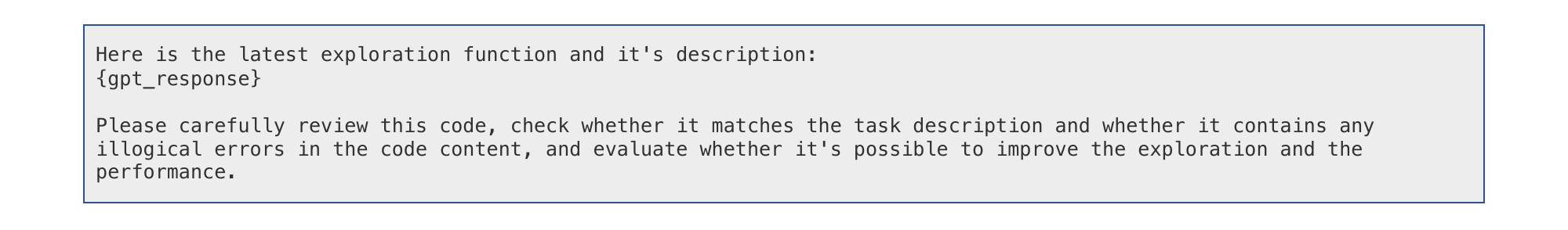}\vspace{-6mm}
\caption{$\texttt{LLM}_g$'s feedback to $\texttt{LLM}_c$.} 
\end{figure}

\begin{figure}[H]
\centering
\includegraphics[width=1\textwidth]{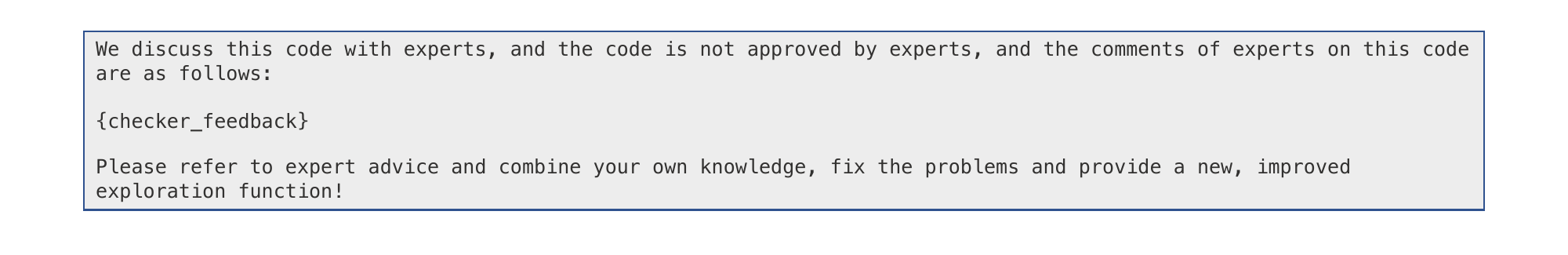}\vspace{-6mm}
\caption{$\texttt{LLM}_c$'s feedback to $\texttt{LLM}_g$.} 
\end{figure}

\begin{figure}[H]
\centering
\includegraphics[width=1\textwidth]{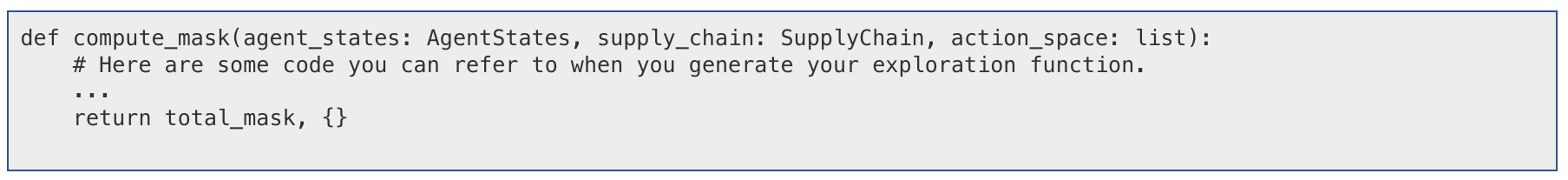}\vspace{-4mm}
\caption{Signature of exploration function.} 
\end{figure}

\begin{figure}[H]
\centering
\includegraphics[width=1\textwidth]{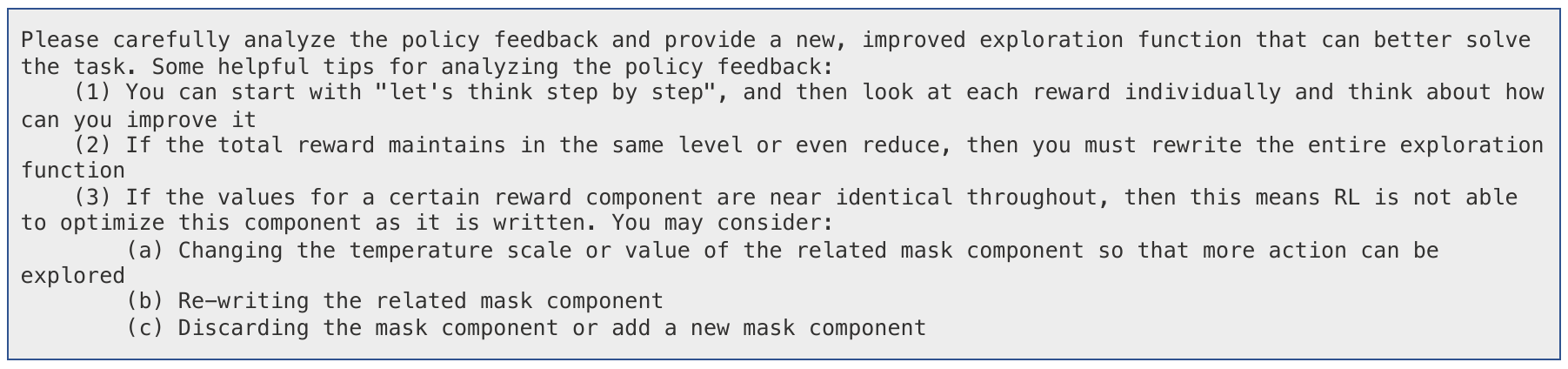}\vspace{-3mm}
\caption{Output and improvement tips for $\texttt{LLM}_g$.} 
\end{figure}

\begin{figure}[H]
\centering
\includegraphics[width=1\textwidth]{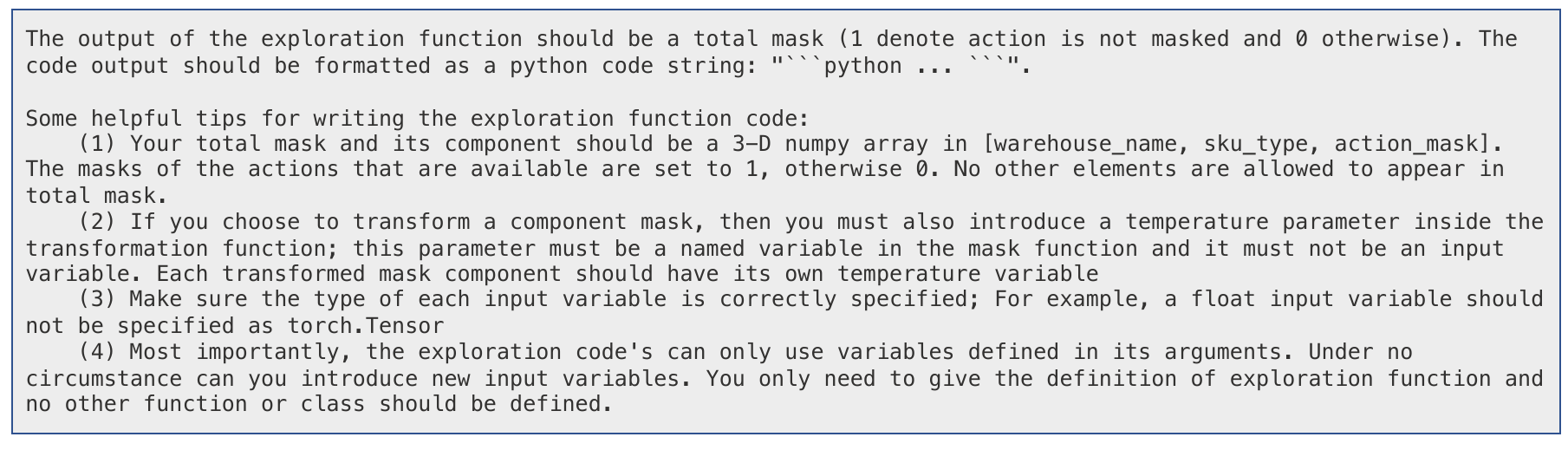}\vspace{-4mm}
\caption{Output format for $\texttt{LLM}_g$.} 
\end{figure}

\begin{figure}[H]
\centering
\includegraphics[width=1\textwidth]{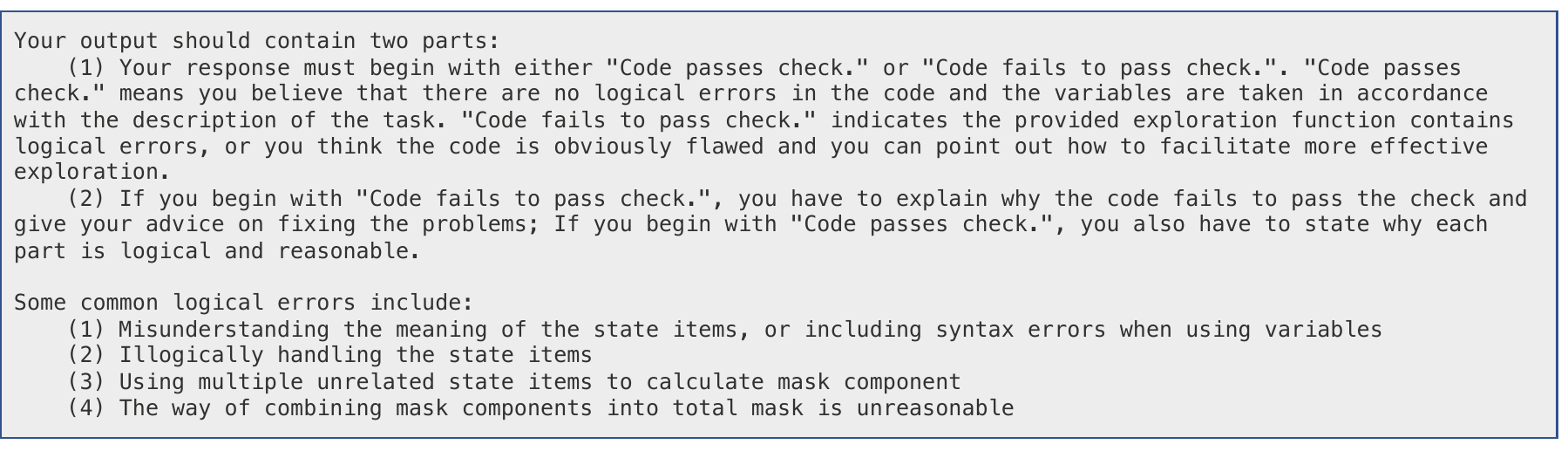}\vspace{-3mm}
\caption{Output format for $\texttt{LLM}_c$.} 
\end{figure}

\begin{figure}[H]
\centering
\includegraphics[width=1\textwidth]{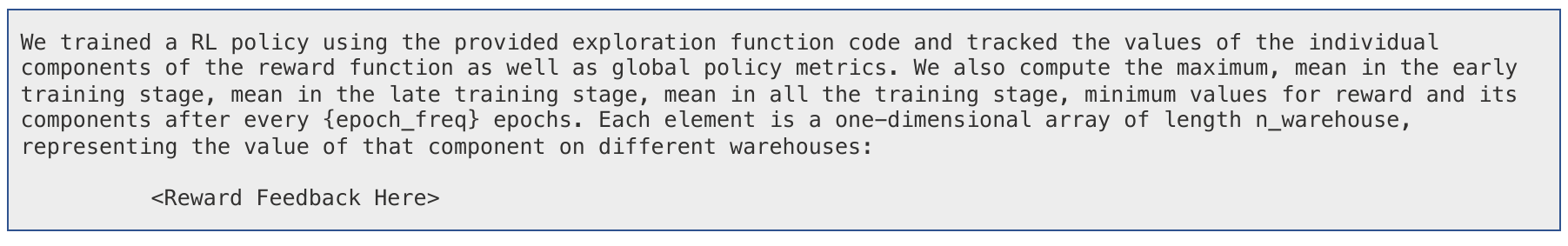}\vspace{-4mm}
\caption{Reward feedback and action feedback.} 
\end{figure}

\newpage

\section{\texttt{eSpark}'s exploration function editing} \label{eSpark improvement}
In this section, we demonstrate the reward editing capabilities of $\texttt{eSpark}$. $\texttt{eSpark}$ is capable of reflecting upon feedback to optimize the exploration for subsequent iterations.
\vspace{-2mm}
\begin{figure}[H]
\centering
\includegraphics[width=0.95\textwidth]{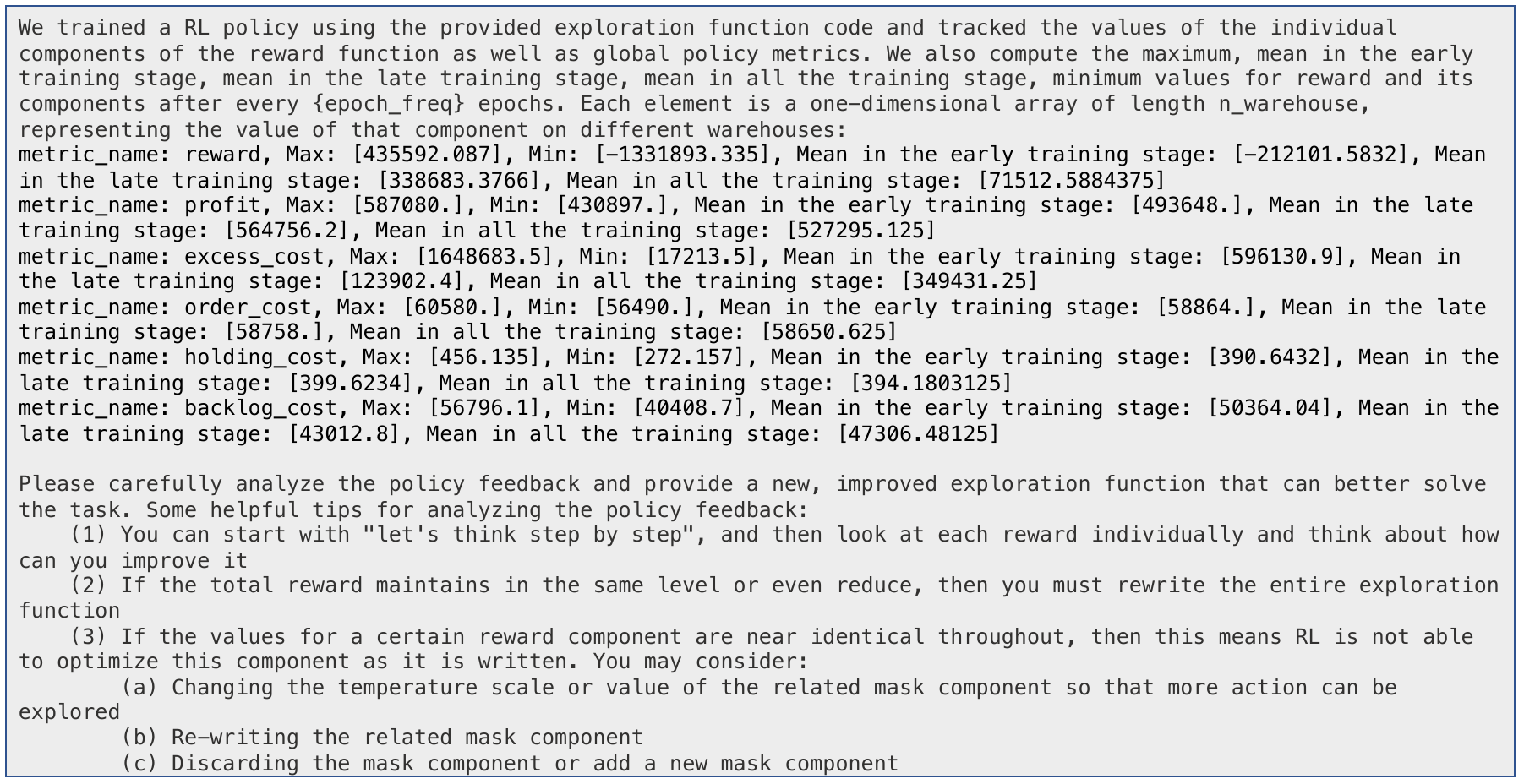}\vspace{-2mm}
\caption{Policy feedback from the last iteration.} 
\end{figure}

\vspace{-2mm}
\begin{figure}[H]
\centering
\includegraphics[width=0.95\textwidth]{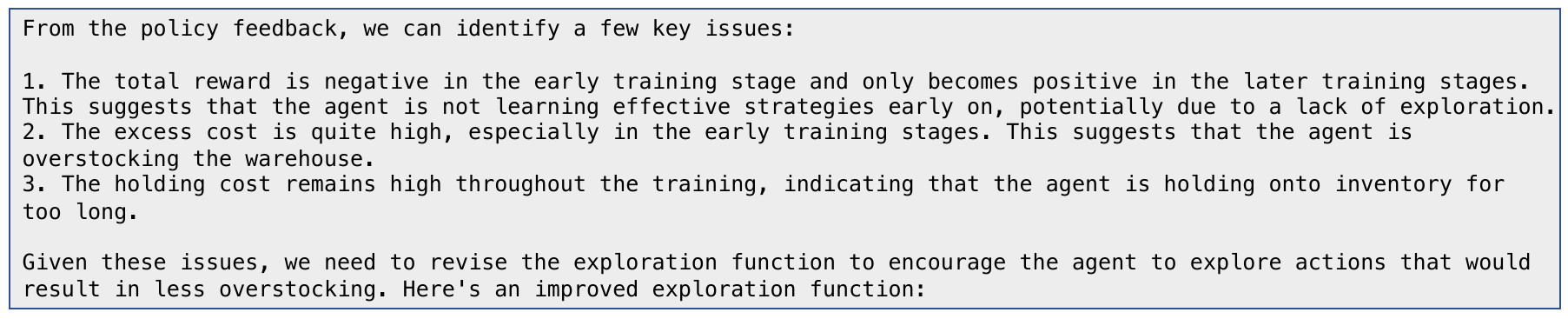}\vspace{-2mm}
\caption{\texttt{eSpark} reflects upon the policy feedback and proposes the modification.} 
\end{figure}

\begin{figure}[H]
\centering
\includegraphics[width=0.9\textwidth]{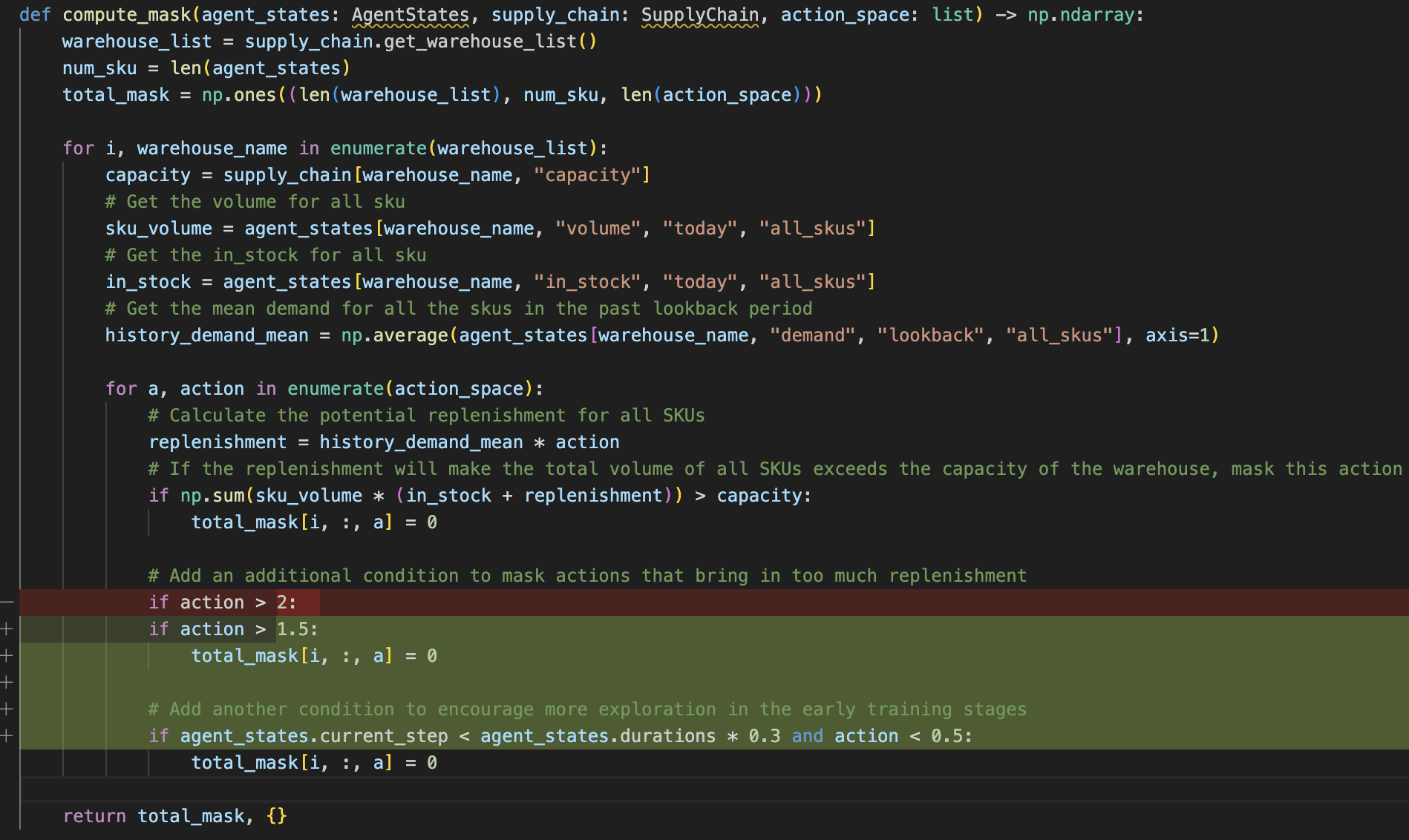}
\caption{Comparison of exploration functions before and after editing.} 
\end{figure}

\section{Impact statement}\label{impact statement}

The integration of MARL into real-world applications holds transformative potential for societal infrastructure, including transportation, logistics, and public service coordination. Our work introduces \texttt{eSpark}, a framework that significantly improves MARL scalability and efficiency by leveraging large language models to prune unnecessary action spaces. This advancement fosters the development of autonomous systems capable of complex, cooperative interactions, without the need for extensive manual tuning. By enabling the management of large-scale agent systems, it can optimize resource allocation, reduce congestion, and enhance service delivery in urban environments. The technology also promises to improve the robustness and reliability of AI systems, ensuring they operate within ethical boundaries and contribute positively to societal norms and values. As we move towards an increasingly automated future, \texttt{eSpark} represents a step towards responsible AI development, prioritizing societal well-being and the advancement of collective, intelligent problem-solving.


\end{document}